\documentclass[peerreview]{IEEEtran}
\usepackage{graphicx,float}
\usepackage{amssymb}
\usepackage{amsmath}
\usepackage{amsfonts}
\usepackage{subfig}
\usepackage{epsfig}

\usepackage{xr}
\usepackage{float}
\usepackage{xspace}
\usepackage{bm}
\usepackage{algorithm,algpseudocode}
\usepackage{epstopdf}

\usepackage[font={small}]{caption}
\usepackage{url}
\usepackage[toc,nonumberlist,acronym,sort=standard]{glossaries}
\usepackage{lipsum,multicol}
\usepackage{xcolor}
\usepackage{dsfont}
\usepackage{soul}
\usepackage[compress]{cite}
\usepackage{linegoal}
\usepackage{array}
\usepackage{soul}
\usepackage{hyperref}
\usepackage{amsthm}
\usepackage[utf8]{inputenc}
\usepackage{multibib}
\newcites{latex}{REFERENCE}

\makeatletter
\newcommand{\algmargin}{\the\ALG@thistlm}
\makeatother
\newlength{\whilewidth}
\algdef{SE}[parWHILE]{parWhile}{EndparWhile}[1]
  {\parbox[t]{\dimexpr\linewidth-\algmargin}{%
     \hangindent\whilewidth\strut\algorithmicwhile\ #1\ \algorithmicdo\strut}}{\algorithmicend\ \algorithmicwhile}%
\algnewcommand{\parState}[1]{\State%
  \parbox[t]{\dimexpr\linewidth-\algmargin}{\strut #1\strut}}

\newlength{\indexwidth}
\algnewcommand{\parNumState}[2]{\State{#1}%
  \parbox[t]{\dimexpr\linewidth-\algmargin-\dimexpr \indexwidth}{\strut #2\strut}}

\algdef{SE}[SUBALG]{Indent}{EndIndent}{}{\algorithmicend\ }%
\algtext*{Indent}
\algtext*{EndIndent}

\newtheoremstyle{mystyle}
  {}
  {}
  {}
  {}
  {\bfseries}
  {:}
  { }
  {}

\theoremstyle{mystyle}
\newtheorem{assumption}{Assumption}
\newtheorem{thm}{Theorem}
\newtheorem{lem}[thm]{Lemma}

\theoremstyle{definition}
\newtheorem{defn}{Definition}

\newtheorem{rem}{Remark}

\makeglossaries
\newacronym{NE}{NE}{Nash equilibrium}
\newacronym{GNE}{GNE}{Generalized Nash Equilibrium}
\newacronym{CNE}{CNE}{coupled Nash equilibrium}
\newacronym{DNE}{DNE}{decoupled Nash equilibrium}
\newacronym{NBS}{NBS}{Nash bargaining solution}
\newacronym{NFG}{NFG}{network formation game}
\newacronym{GraSP}{GraSP}{gradient support pursuit}
\newacronym{ADMM}{ADMM}{alternating direction method of multipliers}
\newacronym{PALM}{PALM}{proximal alternating linearized minimization}
\newacronym{WAC}{WAC}{wide-area control}
\newacronym{LQR}{LQR}{linear-quadratic regulator}
\newacronym{CARE}{CARE}{cross-coupled algebraic Riccati equations}
\newacronym{KKT}{KKT}{Karush-Kuhn-Tucker}
\newacronym{LMI}{LMI}{Linear Matrix Inequality}
\newacronym{LTI}{LTI}{linear time invariant}
\newacronym{LPV}{LPV}{linear parameter varying}
\newacronym{NCS}{NCS}{networked control system}
\newacronym{KL}{KL}{Kurdyka-{\L}ojasiewicz}

\begin{document}
%
\title{Game-Theoretic Mixed $H_2/H_{\infty}$ Control with Sparsity Constraint for Multi-agent Networked Control Systems}
%
%
%
\author{Feier Lian, Aranya Chakrabortty~\IEEEmembership{Senior Member,~IEEE}, and Alexandra Duel-Hallen~\IEEEmembership{Fellow,~IEEE}\thanks{The authors are with the Department of Electrical and Computer Engineering, North Carolina State University, Raleigh, NC 27695 USA (e-mail: flian2@ncsu.edu; aranya.chakrabortty@ncsu.edu; sasha@ncsu.edu)} \thanks{Financial support from the NSF Grant EECS 1544871 is gratefully acknowledged.}}

\glsresetall

\maketitle
\begin{abstract}

Multi-agent \glspl{NCS} are often subject to model uncertainty and are limited by large communication cost, associated with feedback of data between the system nodes. 
To provide robustness against model uncertainty and to reduce the communication cost, this paper investigates the mixed $H_2/H_{\infty}$ control problem for \gls{NCS} under the sparsity constraint.
First, \gls{PALM} is employed to solve the centralized social optimization where the agents have the same optimization objective. Next, we investigate a sparsity-constrained noncooperative game, which accommodates different control-performance criteria of different agents,
 and propose a best-response dynamics algorithm based on \gls{PALM} that converges to an approximate \gls{GNE} of this game. A special case of this game, where the agents have the same $H_2$ objective, produces a partially-distributed social optimization solution.
We validate the proposed algorithms using a network with unstable node dynamics and demonstrate the superiority of the proposed \gls{PALM}-based method to a previously investigated sparsity-constrained mixed $H_2/H_{\infty}$ controller.

\end{abstract}


%
\IEEEpeerreviewmaketitle

\section{Introduction}
\glsresetall

Recent research on multi-agent networks has proposed various methods for reducing communication cost for control by using sparse $H_2$ control designs \cite{mihailo,lian2017game,javad,dorfler2014sparsity,lin2017co,matni2016regularization}. However, most of this work ignores the effects of model uncertainties, which are bound to arise in most practical large-scale systems since the network operating conditions and topology change frequently over time. Even if the topology is fixed, the exact model parameters are not always available to the designer.
The sparse control design in such cases must also be robust against the uncertainties. Robust designs have been reported in several recent papers such as \cite{lidstrom2016optimal,arastoo2016closed,bahavarnia2017sparse,bahavarnia2017state} using $H_{\infty}$ control, which is suitable for handling norm-bounded uncertainties in the system dynamics. In particular, \cite{arastoo2016closed,bahavarnia2017sparse} employ both $H_2$ and $H_{\infty}$ control, thus balancing the $H_2$ performance of the nominal system and robustness objective. However, optimizing the $H_2$ performance under the $H_{\infty}$ and sparsity constraints has not been investigated by other researchers. Moreover, while a global control performance cost was optimized, this metric did not address the individual objectives of multiple agents under uncertainty.

Control of NCS under uncertainties has received significant attention 
recently in various domains, such as wide-area control of power systems \cite{lian2017game}, multi-robot coordination, multi-access broadcast channel, vehicle formation, and wireless sensor networks\cite{seuken2008formal,ogren2004cooperative}. In these problems, game theory becomes a powerful tool, with different control inputs modeled as game players, where each player aims to optimize its individual objective using an associated control policy. Differential games have been investigated for uncertain multi-agent systems, and algorithms for finding an equilibrium point have been proposed based on solving a set of coupled optimization problems. 
The works \cite{jungers2008bounded,de2017finite} extend Nash-type differential game in \cite{basar85} by finding robust Nash strategies while either considering polytopic uncertainty or formulating uncertain external disturbance as a fictitious player. 
The works \cite{mukaidani2009robust,mukaidani2014h,mukaidani2015stackelberg,mukaidani2016dynamic} model the uncertainty using stochastic differential equations, and the Nash strategies are found by solving cross-coupled matrix equations through necessary optimality conditions or \Gls{KKT} conditions. 
Recently, reinforcement learning has been applied to multi-agent control Nash games when the system parameters are completely or partially unknown \cite{vrabie2010integral, vamvoudakis2015non,song2017off,vamvoudakis2017game}. However, these reported game-theoretic designs did not address any sparsity constraint.

In this paper, we investigate controller designs for multi-agent systems that aim to reduce $H_2$ cost under $H_{\infty}$ and sparsity constraints. Both social optimization and a nonconvex game, where the $H_2$-objectives of the agents are same and different, respectively, are investigated. We model our uncertainty as a norm-bounded parameteric uncertainty that translates into an $H_{\infty}$ constraint as in \cite{Lian:2018aa,bahavarnia2017sparse}. We employ the \gls{PALM} \cite{lin2017co}, which has been shown to be effective for optimization for nonconvex nonsmooth problems \cite{bolte2014proximal} and was utilized in \cite{lin2017co} in a sparsity-constrained output-feedback co-design problem.
First, a centralized sparsity-constrained mixed $H_2/H_{\infty}$ controller, which represents the social optimization, is addressed.
Preliminary results on this topic were recently reported in our conference paper \cite{Lian:2018aa}, where we developed a centralized controller under the sparsity and $H_{\infty}$ constraints using a greedy \gls{GraSP} method \cite{bahmani2013greedy}. However, the algorithm in \cite{Lian:2018aa} requires the knowledge of an initial stabilizing feedback gain that satisfies the sparsity constraint. We eliminate this requirement and show the advantages of the \gls{PALM}-based controller in this paper over that in \cite{Lian:2018aa} in terms of the quadratic $H_2$-cost.


Next, we extend the proposed design to the multi-agent scenario where each agent designs its own part of the feedback matrix, subject to a shared global $H_{\infty}$-norm and sparsity constraints. Since each agent has different individual cost, the control design is modeled as a noncooperative game with shared constraints. We develop a numerical algorithm to find the \gls{GNE} \cite{paccagnan2016distributed} of this game. The proposed algorithm has partially-distributed computation, i.e., in the first stage, each player computes its own feedback matrix while in the second stage, the sparse links are chosen globally based on the results from the first stage. Third, assuming all players of the game have the same $H_2$-optimization objective, we develop a potential game that yields a partially-distributed implementation of the social optimization. We perform numerical simulations to demonstrate the performance of the proposed algorithms and discuss their convergence properties.

\begin{table}[!t]
\centering
\smallskip 
\caption{Notation}
\begin{tabular}{p{0.15\columnwidth}p{0.75\columnwidth}}
\hline
{\bf Term} & {\bf Definition}\\
\hline
$\mathbf M {\succ}0({\succeq}0) $ & Matrix $\mathbf M$ is positive definite (semidefinite)\\
\hline
$\mathbf M {\prec}0({\preceq} 0)$ & $\mathbf M$ is negative definite (semidefinite)\\
\hline
$\sigma_{\mathrm{max}}(\mathbf M)$ & maximum singular value of $\mathbf M$\\
\hline
$||{\mathbf K}||_F$ & Frobenius norm of the matrix ${\mathbf K}$, defined by $\sqrt{\mathrm{trace}({\mathbf K}^{\mathrm{T}}{\mathbf K})}$. \\
\hline
$\mathrm{card}({\mathbf K})$ & Cardinality of matrix $\mathbf K$, defined by the number of nonzero elements in $\mathbf K$.\\
\hline
$\nabla_{\mathbf K}J({\mathbf K})$ & The gradient of the scalar function $J({\mathbf K})$ with respect to the matrix ${\mathbf K}$. Assuming ${\mathbf K}\in \mathbb{R}^{m\times n}$, $\nabla_{\mathbf K}J({\mathbf K})$ is given by a $m\times n$ matrix with the elements $[\nabla_{\mathbf K}J({\mathbf K})]_{ij} = \partial J/ \partial K_{ij}$ .\\ 
\hline
$[{\mathbf K}]_s$ & The matrix obtained by preserving only the $s$ largest-magnitude entries of the matrix ${\mathbf K}$ and setting all other entries to zero. \\
\hline 
$H_2$ norm & The system $H_2$ norm in the time domain is $||G||_2=\left(\int_0^{\infty}{\left[\mathrm{trace}(H(t)^TH(t))\right]dt}\right)^{1/2}$, where $H(t)$ is the impulse response of the system.\\
\hline
$H_{\infty}$ norm & The system $H_{\infty}$ norm $||G||_{\infty}={\mathrm{sup}}_{\omega}~\sigma_{\mathrm{max}}(G(j\omega))$, where $G(j\omega)$ is the system transfer matrix.\\
\hline
\end{tabular}
\label{palm-notation:tb}
\end{table}
 
The main contributions of the paper can, therefore, be summarized as:
\begin{itemize}
\item {Development and analysis of a centralized, sparsity-constrained mixed $H_2/H_{\infty}$ controller for social optimization of multi-agent systems with norm-bounded uncertainty.}
\item {Development of game-theoretic, partially-distributed algorithms that aim to minimize the $H_2$-norms of the agents' transfer functions under shared sparsity and $H_{\infty}$-norm constraints. }
\end{itemize}

The rest of the paper is organized as follows. Section \ref{palm-centra:sec} presents the system model with parametric uncertainty and develops a centralized \gls{PALM} algorithm for social optimization using sparsity-constrained mixed $H_2/H_{\infty}$ control. Section \ref{palm-game:sec} describes a multi-agent system with parametric uncertainty, proposes a noncooperative game with shared sparsity and $H_{\infty}$ constraints, and develops partially-distributed numerical algorithms for this game and for social optimization.
Section \ref{palm-num:sec} demonstrates effectiveness of the proposed algorithms using numerical simulations, and discusses their complexities and convergence properties. Section \ref{palm-concl:sec} discusses future directions and concludes the paper.

Throughout the paper, matrices are denoted with boldface capital letters. If $\mathbf M$ is a symmetric matrix, the upper block matrices are sometimes denoted by $*$ to save space. Some notation used is summarized in Table \ref{palm-notation:tb}.

\section{\gls{PALM} algorithm for centralized sparsity-constrained mixed $H_2/H_{\infty}$ control}
\label{palm-centra:sec}

\subsection{System model and mixed $H_2/H_{\infty}$ control}

Consider the following linear time-invariant system with model uncertainty:
\begin{align}
\dot{\mathbf{x}}(t) &= \underbrace{(\mathbf A+\Delta \mathbf A)}_{\hat{\mathbf A}}\mathbf x(t) + \underbrace{(\mathbf B+\Delta \mathbf B)}_{\hat{\mathbf B}}\mathbf u(t) + \mathbf B_2 \mathbf w_2(t)\nonumber\\
\mathbf{z}_2(t) &= \mathbf C_2 \mathbf{x}(t) + \mathbf D_2 \mathbf{u}(t) + \mathbf D_{22} \mathbf w_2(t)\nonumber\\
\mathbf y(t) &= \mathbf C \mathbf x(t),
\label{uncertain:eq}
\end{align}
where $\mathbf x(t)\in \mathbb{R}^{n\times 1}$ is the state vector, $\mathbf u(t) \in \mathbb{R}^{m\times 1}$ is the control input vector, $\mathbf{w}_2(t) \in \mathbb{R}^{m_2 \times 1}$ is the exogenous input, $\mathbf{z}_2(t)\in \mathbb{R}^{p_2\times 1}$ is the performance output, and $\mathbf y(t){\in} \mathbb{R}^{p \times 1}$ is the measured output.
\noindent The matrices $\mathbf A$ and $\mathbf B$ are the nominal values of the state and input matrices, respectively,  while $\Delta \mathbf A$ and $\Delta \mathbf B$ model the respective uncertainties. We make the following assumptions:
\begin{assumption}
\label{controllable:assumption}
(i) The pair $(\mathbf A, \mathbf B)$ is controllable, $(\mathbf C, \mathbf A)$ is observable,  $(\mathbf C_2, \mathbf A)$ is observable.\\
(ii) $\Delta \mathbf A$ and $\Delta \mathbf B$ have the form \cite{bahavarnia2017sparse}
\begin{equation}
[\Delta \mathbf A  ~\Delta \mathbf B]= \mathbf B_1 \Delta \boldsymbol \delta [\mathbf C_1 ~\mathbf D_1], \label{deltaA:eq}
\end{equation}
where $\mathbf B_1{\in} \mathbb{R}^{n\times m_1}$, $\mathbf C_1{\in} \mathbb{R}^{p_1\times n}$, $\mathbf D_1 {\in} \mathbb{R}^{p_1\times m}$ are known matrices, and $\Delta \boldsymbol \delta {\in} \mathbb{R}^{m_1\times p_1}$ is an unknown matrix which is norm-bounded, satisfying $\Delta \boldsymbol\delta^T\Delta \boldsymbol \delta \preceq 1/\gamma^2\mathbf I$ for any scalar $\gamma>0$.
\end{assumption}
\begin{assumption}
\label{bounded:assumption}
Matrices $\mathbf C_2$ and $\mathbf D_2$ have the following form:
\begin{equation}
\label{palm-h2:eq}
\mathbf{C}_2 = \begin{bmatrix}\mathbf{C}_{2}^1, \mathbf{0} \end{bmatrix}, \mathbf{D}_2 = \begin{bmatrix} \mathbf{0} \\ \mathbf{D}_{2}^2 \end{bmatrix}, \mathbf{C}_2^T\mathbf{D}_2=0. 
 \end{equation}
\end{assumption}

Using assumptions \ref{controllable:assumption} and \ref{bounded:assumption}, the system (\ref{uncertain:eq}) can be expressed as the feedback interconnection of the following two subsystems:
\begin{align}
\label{sigma:eq}
&\Sigma: \begin{cases}
\dot{\mathbf x}(t) = \mathbf A \mathbf x(t) + \mathbf B\mathbf u(t) + \mathbf B_1\mathbf w_1(t) + \mathbf B_2\mathbf w_2(t) \\
\mathbf z_1(t) = \mathbf{C}_1\mathbf{x}(t) + \mathbf D_1 \mathbf u(t)\\
\mathbf z_2(t) = \mathbf C_2\mathbf x(t) + \mathbf D_2\mathbf u(t) + \mathbf D_{22} \mathbf w_2(t)\\
\mathbf y(t) = \mathbf C\mathbf x(t)\\
\end{cases}\raisetag{1\baselineskip}
\end{align}
\begin{align}
&\Sigma_K: \begin{cases}
\mathbf w_1(t) = \Delta \boldsymbol \delta\mathbf z_1(t),
\end{cases}
\label{sigmak:eq}
\end{align}
where $\mathbf z_1(t)\in \mathbb{R}^{p_1\times 1}$, $\mathbf w_1(t)\in \mathbb{R}^{m_1\times 1}$.

Our goal in this section is to find a linear static output-feedback controller $\mathbf u(t) = -\mathbf K \mathbf y(t)$ that stabilizes (\ref{uncertain:eq}), i.e., guarantees $||T_{z_1w_1}(\mathbf{K}) ||_{\infty} < \gamma$.  Following \cite{kami2008gradient}, this $H_{\infty}$-norm constraint can be transformed into an \gls{LMI} condition as stated in the following theorem.

\begin{thm}
The $H_{\infty}$-norm constraint $||T_{z_1w_1}(\mathbf K)||_{\infty} < \gamma$ holds if and only if there exists an $\mathbf{X}=\mathbf X^T$ that satisfies the \gls{LMI}
\begin{align}
& \begin{bmatrix}
\mathbf{A}_{cl}(\mathbf K)\mathbf{X} + \mathbf{X}\mathbf{A}_{cl}(\mathbf K)^T+\mathbf{B}_1\mathbf{B}_1^T & \mathbf{XC}_{cl1}(\mathbf K)^T\\
\mathbf{C}_{cl1}(\mathbf K)\mathbf{X} & -\gamma^2 \mathbf{I}
\end{bmatrix} \prec 0 \nonumber\\
& \mathbf{X} \succ 0
\label{LMI_inf:eq}
\end{align}
\noindent where $\mathbf{A}_{cl}(\mathbf K)= \mathbf{A}-\mathbf{BKC}$, $\mathbf{C}_{cli}(\mathbf K) = \mathbf{C}_i - \mathbf{D}_i\mathbf{KC}$ for $ i=1,2$. 
\end{thm}



The {\bf mixed $H_2/H_\infty$ control problem} can then be stated as: Given an achievable $H_{\infty}$-norm bound $\gamma$, find a feedback controller $\mathbf K \in \mathbb{R}^{m \times p} $ that solves
\begin{eqnarray}
\underset{\mathbf K}{\mathrm{Minimize}~}||T_{z_2w_2}(\mathbf K)||_2, \label{p1}\\
\mbox{s.t}\;\; \mathbf u(t) = -\mathbf K \mathbf y(t)), \mbox{equation (4) holds, and} \label{p2} \\
||T_{z_1w_1}(\mathbf{K}) ||_{\infty} < \gamma\; (\mbox{or equivalently, (6) holds}). \label{p3}
\end{eqnarray}
\noindent $T_{z_2w_2}$ is the closed-loop transfer function from $\mathbf{w}_2$ to $\mathbf{z}_2$. 

For simplicitly, and without loss of generality, we set $\mathbf D_{22}=\mathbf 0$ in (\ref{uncertain:eq}). Following standard robust control results, such as in \cite{kami2008gradient}, it can be shown that the squared $H_2$ norm from $\mathbf w_2$ to $\mathbf z_2$ for the system (\ref{sigma:eq}) is
\noindent
\begin{equation}
||T_{z_2w_2}(\mathbf K)||_2^2 = J(\mathbf{K}):= \mathrm{trace}(\mathbf{B}_2^{T}\mathbf{P}\mathbf{B}_2)
\label{Jk:eq}
\end{equation}
\noindent where $\mathbf{P}$ is the solution of the Lyapunov equation
\begin{equation}
\mathbf{PA}_{cl}(\mathbf K) + \mathbf{A}_{cl}(\mathbf K)^{T}\mathbf{P}+\mathbf{C}_{cl2}(\mathbf K)^T\mathbf{C}_{cl2}(\mathbf K) = 0.
\end{equation}
\noindent 
We can define
\begin{equation}
\label{qr:eq}
 \mathbf{Q}=(\mathbf{C}_2^1)^T\mathbf{C}_2^1 \succeq 0,~ \mathbf{R}=(\mathbf{D}_2^2)^T\mathbf{D}_2^2 \succ 0
 \end{equation}
in which case the objective $J(\mathbf K)$ in (\ref{Jk:eq}) can be written as
\begin{align}
\label{PALM-J-centra:eq}
J(\mathbf K) = \int_{t=0}^{\infty}{\left[\mathbf x^T(t)\mathbf Q\mathbf x(t)+\mathbf u^T(t)\mathbf R\mathbf u(t) \right]dt}.
\end{align}

\subsection{Sparsity-constrained mixed $H_2/H_{\infty}$ control}

The solution $\mathbf K$ in problem \eqref{p1}-\eqref{p3} is usually a dense matrix, meaning that every sensor must send the sensed outputs to every controller. This can result in a large communication cost.
To reduce this cost, we impose a sparsity constraint on the feedback matrix \cite{lian2017game,Lian:2018aa}, resulting in the following sparsity-constrained mixed $H_2/H_{\infty}$ problem: 
\begin{align}
\underset{\mathbf K}{\mbox{min }}& ||T_{z_2w_2}(\mathbf K)||_2,  \nonumber \\ 
 \mbox{ s.t. } & ||T_{z_1w_1}(\mathbf{K}) ||_{\infty} < \gamma, ~ \mathrm{card}(\mathbf K) \leq s,
\label{P}
\end{align}
with the plant model satisfying (4). For simplicity, we define each nonzero entry in $\mathbf{K}$ as one communication link. Alternative definitions of sparsity and their effects on the actual cost of communication are discussed in \cite{lian2017game}.

\subsection{Overview of the centralized PALM algorithm}
\label{palm-centr:sec}

The sparsity-constrained mixed $H_2/H_{\infty}$ problem (\ref{P}) was solved in our recent paper \cite{Lian:2018aa} using the \gls{GraSP} algorithm, assuming that for any given value of $s$ we can find an initial guess for $\mathbf K$ that satisfies the $s$-sparse structure. Depending on the plant model and the uncertainty in (\ref{sigma:eq}), finding such a feasible initial guess in reality, however, can be quite difficult. In this section, we eliminate this requirement by introducing a sparsity-constrained optimization algorithm based on \gls{PALM}.
For this, we first transform \eqref{P} into a problem with two optimization variables, $\mathbf K$ and $\mathbf F$, where $\mathbf K$ is defined in (\ref{p2}),
and $\mathbf F$ represents the sparse feedback matrix that satisfies the cardinality constraint. The problem (\ref{P}) can then be reformulated as follows:
\begin{align}
\underset{\mathbf K, \mathbf F}{\mbox{min }}& J(\mathbf K) + \frac{\rho}{2} || \mathbf K - \mathbf F||_{F}^2,  \nonumber \\ 
 \mbox{ s.t. } & ||T_{z_1w_1}(\mathbf{K}) ||_{\infty} < \gamma, \nonumber \\
&\mathrm{card}(\mathbf F) \leq s,
\label{Pb}
\end{align}
\noindent where the penalty term $\rho/2\vert|\mathbf K - \mathbf F \vert_F^2$ is used to regularize the difference between $\mathbf K$ and $\mathbf F$. When the parameter $\rho$ is chosen large enough, this term can be reduced sufficiently.
There are two constrained variables in (\ref{Pb}). We next transform (\ref{Pb}) to an unconstrained optimization problem by defining the following indicator functions:
\begin{align}
g(\mathbf K) = \left\{\begin{array}{cl}
0, & T_{\infty}(\mathbf K) < \gamma\\
+\infty, & O.W.
\end{array}\right.
\label{palm-gK:eq}
\end{align}
\begin{align}
f(\mathbf F) = \left\{\begin{array}{cl}
0, & \mathrm{card}(\mathbf F) \leq s\\
+\infty, & O.W.
\end{array}\right.
\label{palm-fF:eq}
\end{align}
Using (\ref{palm-gK:eq}) and (\ref{palm-fF:eq}) the problem (\ref{Pb}) can be written as
\begin{align}{}
\underset{\mathbf K, \mathbf F}{\mbox{min }}& \Phi(\mathbf K, \mathbf F),  
\label{Pc}
\end{align}
\noindent where
\begin{equation}
\Phi(\mathbf K, \mathbf F)= J(\mathbf K) + g(\mathbf K) + f(\mathbf F) + H(\mathbf K, \mathbf F)
\label{Phi:eq}
\end{equation}
\noindent with
\begin{equation}
\label{H:eq}
H(\mathbf K, \mathbf F) = \frac{\rho}{2}|| \mathbf K - \mathbf F||_F^2
\end{equation}
\noindent being the coupling function between $\mathbf K$ and $\mathbf F$.

The \gls{PALM} algorithm proceeds by alternating the minimization on the variables $(\mathbf K, \mathbf F)$ through separate subproblems \cite{bolte2014proximal}, which simplifies solving (\ref{Pc}), as described below.
 When $\mathbf K$ is fixed, the optimization (\ref{Pc}) reduces to minimizing the sum of a nonsmooth function $f(\mathbf F)$ and a smooth function $H(\mathbf K, \mathbf F)$ of $\mathbf F$. From the result of proximal forward-backward splitting algorithm \cite{bolte2014proximal}, minimizing $f + H$ can be relaxed as iteratively upper bounding the objective and minimizing the upper bound \cite{parikh2014proximal}. The iteration can be written as
\begin{align}
\label{majorization:eq}
&\mathbf F^{k+1} = \\
&\arg\min_{\mathbf F} \left\{ \langle
 \mathbf F - \mathbf F^k , \nabla_{\mathbf F}H(\mathbf F, \mathbf K) \rangle  + \frac{t}{2}\vert| \mathbf F - \mathbf F^k \vert|_F^2 + f(\mathbf F) \right\},\nonumber
\end{align}
\noindent where $\langle, \rangle$ denotes inner product. Minimizing the first two terms is equivalent to minimizing the first order (linear) approximation of $H(\mathbf K, \mathbf F)$ at $\mathbf F = \mathbf F^k$, regularized by a trust-region penalty near $\mathbf F^k$. When $t \in (L, \infty)$ and $L$ is the Lipschitz constant (see S.II \cite{Lian:aa} and Appendix B.3 \cite{Lian:2019aa}) for $\nabla_{\mathbf F}H(\mathbf K, \mathbf F)$, the regularized linear approximation provides an upper bound on $H(\mathbf K, \mathbf F)$ \cite{parikh2014proximal}.

Eq. (\ref{majorization:eq}) can be rewritten compactly using the definition of a proximal map as
\begin{equation}
\mathbf F^{k+1} \in \mathrm{prox}_t^f \left( \mathbf F^k - 1/t\nabla_{\mathbf F} H(\mathbf K, \mathbf F^k) \right),
\end{equation} 
\noindent where for $\sigma: \mathbb{R}^d \rightarrow (\infty, \infty]$, a proper and lower semicontinuous function, $x\in \mathbb{R}^d$ and $t>0$, the proximal map associated with $\sigma$ at point $\mathbf x$ is
\begin{equation}
\label{proximal-map:eq}
\mathrm{prox}_t^{\sigma}(\mathbf x) = \arg\min_{\mathbf u \in \mathbb{R}^d} \left\{\sigma(\mathbf u) + \frac{t}{2}\vert| \mathbf u - \mathbf x \vert|^2 \right\}.
\end{equation}

Similar analysis can be carried out for the minimization of (\ref{Pc}) when $\mathbf F$ is fixed.
In summary, the \gls{PALM} algorithm minimizes (\ref{Pc}) by alternatively finding the proximal maps:
\begin{eqnarray}
\label{alternate-proximal-F:eq}
\mathbf F^{k+1} &\in& \mathrm{prox}_{a_k}^f \left( \mathbf F^k - 1/a_k\nabla_{\mathbf F} H(\mathbf K^k, \mathbf F^k) \right) \\
\mathbf K^{k+1} &\in& \mathrm{prox}_{b_k}^{J+g} \left( \mathbf K^k - 1/b_k\nabla_{\mathbf K} H(\mathbf K^k, \mathbf F^{k+1}) \right)
\label{alternate-proximal-K:eq}
\end{eqnarray}
where $a_k$ and $b_k$ are positive constants that are greater than the Lipschitz constants $L_1(\mathbf K^k)$ and $L_2(\mathbf F^{k+1})$ of $\nabla_{\mathbf F} H(\mathbf K^k, \mathbf F)$ and $\nabla_{\mathbf K} H(\mathbf K, \mathbf F^{k+1})$, respectively. 

\begin{algorithm}[!b]
     \caption{PALM algorithm for the mixed $H_2/H_{\infty}$ control algorithm with sparsity constraint}
   \label{palm-alg1:alg}
  \begin{algorithmic}
  \State {\bf Given} $s$: sparsity constraint, $\gamma$: $H_{\infty}$-norm bound.
  \State  {1. \bf Initialization}: 
  \State $\mathbf K^0$: any stabilizing feedback gain with $T_{\infty}(\mathbf K^0) < \gamma$.
  \State $\mathbf F^0$: any stabilizing feedback gain $\mathbf F^0$. 
  \State Compute $a:= \gamma_1\rho$, $b := \gamma_2 \rho$.
  \For{$k=1,2,...k_{\mathrm{max}}$ until  $||\mathbf K^{k+1}-\mathbf K^k||_{F} < \epsilon_1$ or $||\mathbf F^{k+1}-\mathbf F^k||_{F} < \epsilon_2$}
  \State // 2. $\mathbf F$-minimization step
  \State 2.1 Compute $\mathbf Z^k := \mathbf F^k - \frac{1}{a}\nabla_{\mathbf F}H(\mathbf K^k, \mathbf F^k)$
  \State 2.2 Prune $\mathbf Z^k$: $\mathbf F^{k+1}:=[\mathbf Z^k]_s$.
  \State // 3. $\mathbf K$-minimization step
  \State 3.1 Compute $\mathbf X^k := \mathbf K^k - \frac{1}{b} \nabla_{\mathbf K} H(\mathbf K^k, \mathbf F^{k+1})$.
  \State 3.2 Update $\mathbf K^{k+1}$: $\mathbf K^{k+1} := \textsc{KproxOp}(\mathbf K^k, \mathbf X^k, b)$.
  \EndFor
   \end{algorithmic}
\end{algorithm}

\subsection{Algorithm description}

We summarize the \gls{PALM} algorithm for sparsity-constrained mixed $H_2/H_{\infty}$ control in Algorithm \ref{palm-alg1:alg}.
In Steps 2 and 3 of Algorithm \ref{palm-alg1:alg}, \textit{$\mathbf F$-minimization} (\ref{alternate-proximal-F:eq}) and \textit{$\mathbf K$-minimization} (\ref{alternate-proximal-K:eq}) are performed, respectively.
In Step 2, we perform iterative \textit{$\mathbf F$-minimization} (\ref{alternate-proximal-F:eq}), which can be rewritten from  (\ref{proximal-map:eq}) as:
\begin{align}
\label{F_opr:eq}
\mathbf F^{k+1} & = \arg\min_{\mathbf F}\left\{ f(\mathbf F) + \frac{a_k}{2}|| \mathbf F - \mathbf Z^k||_F^2 \right\}
\end{align}
\noindent where $\mathbf Z^k$ is the point within the parenthesis in (\ref{alternate-proximal-F:eq}), found in Step 2.2.
\noindent It is easy to see that the partial gradients of $H(\mathbf K, \mathbf F)$, defined in (\ref{H:eq}) are
\begin{align}
\label{gradH:eq}
\nabla_{\mathbf K} H(\mathbf K, \mathbf F) & = \rho(\mathbf K - \mathbf F)\nonumber \\
\nabla_{\mathbf F} H(\mathbf K, \mathbf F) & = \rho(\mathbf F - \mathbf K).
\end{align}
\noindent From (\ref{gradH:eq}), the Lipschitz constant $L_1(\mathbf K^k)=\rho$, and thus the constant $a_k$ in (\ref{alternate-proximal-F:eq}) and (\ref{F_opr:eq}) is defined as
\begin{equation}
a_k = a =\gamma_1\rho
\end{equation}
with $\gamma_1>1$. 



In Step 3 of Algorithm \ref{palm-alg1:alg}, we perform iterative $\mathbf K$-minimization:
\begin{align}
\label{K_opr:eq}
\mathbf K^{k+1} & = \arg\min_{\mathbf K} \left\{ J(\mathbf K) + g(\mathbf K) + \frac{b_k}{2} \vert| \mathbf K - \mathbf X^k \vert|_{F}^2 \right\},
\end{align}
\noindent which is equivalent to (\ref{alternate-proximal-K:eq}), and $b_k$ is chosen as
$b_k = b = \gamma_2\rho$,  with $\gamma_2>1$.



In the following, we present the solutions for Eq. (\ref{F_opr:eq}) and (\ref{K_opr:eq}) used in Steps 2 and 3 of Algorithm \ref{palm-alg1:alg}.

\subsubsection{$\mathbf F$-minimization}
Applying the proximal operator (\ref{alternate-proximal-F:eq}) of function $f$ is equivalent to minimizing a regularized version of $f$. In (\ref{F_opr:eq}), $f$ is an indicator function of the set $\mathcal{X}=\{\mathbf F\vert \mathrm{card}(\mathbf K)\leq s\}$ (\ref{palm-fF:eq}), so the proximal operator in (\ref{F_opr:eq}) (Step 2.3 in Algorithm \ref{palm-alg1:alg}) is equivalent to the projection of $Z^k$ onto the set $\mathcal{X}$. Therefore, we can rewrite (\ref{F_opr:eq}) as
\begin{align}
\label{Fmin:eq}
\mathbf F^{k+1} = &\arg\min_{\mathbf F} ~ \vert| \mathbf F - \mathbf Z^k \vert|_F^2 \nonumber\\
 ~ & \mbox{s.t. }  ~ \mathrm{card}(\mathbf F) \leq s.
\end{align}
As shown in \cite{bolte2014proximal,lin2017co}, the solution to (\ref{Fmin:eq}) is $[\mathbf Z^k]_s$ (see Table \ref{palm-notation:tb}), which is Step 2.3 of Algorithm \ref{palm-alg1:alg}.

\begin{algorithm}[!b]
     \caption{\textsc{KproxOp}: Subroutine to solve (\ref{Kmin:eq})}
   \label{palm-Kmin:alg}
  \begin{algorithmic}[1]
    \Procedure{KproxOp}{$\mathbf K^{\mathrm{cur}}$, $\mathbf X^k$, $b$}
    \While{True}
    \State {$\mathbf K^{\mathrm{prev}} := \mathbf K^{\mathrm{cur}}$ }
    \If {$\vert|\nabla_{\mathbf K} h(\mathbf K^{\mathrm{cur}})\vert|_F < \epsilon_3$}
    \State{// Stationary point in the interior of the $H_{\infty}$-constraint set.}
    \State{break}
    \EndIf
   \State {// Take a gradient-descent step in the interior of the $H_{\infty}$-constraint set.}
   \State {$\mathbf{K}^{\mathrm{cur}} := \mathbf{K}^{\mathrm{prev}} - d\nabla_{\mathbf K} h(\mathbf K^{\mathrm{prev}})$, where step size $d>0$ is chosen by backtracking line search \cite{luenberger1984linear} s.t. $T_{\infty}(\mathbf{K}^{\mathrm{cur}}) < \gamma$}
  \If {$d < \epsilon_2$ }
  \State // $\mathbf K^{\mathrm{prev}}$ is near the boundary of the $H_{\infty}$-constraint set
  \State{Solve for $\mathbf K^{\mathrm{in}}$ using (\ref{fd:eq}). Let $\Delta \mathbf K^{\mathrm{cur}} :=\mathbf K^{\mathrm{in}}-\mathbf K^{\mathrm{prev}}$.}
  \State{$\mathbf{K}^{\mathrm{cur}} := \mathbf{K}^{\mathrm{prev}}+  d' \Delta \mathbf{K}^{\mathrm{cur}}$, where $d'$ is determined by backtracking line search \cite{luenberger1984linear}.}
  \EndIf
  \If {$\vert| \mathbf K^{\mathrm{cur}} - \mathbf K^{\mathrm{prev}}\vert|_F < \epsilon_1$}
  \State break
  \EndIf
  \EndWhile
  \EndProcedure
   \end{algorithmic}
\end{algorithm}

\subsubsection{$\mathbf K$-minimization}

Next we focus on the proximal operator for (\ref{K_opr:eq}), which is equivalent to
\begin{align}
\label{Kmin:eq}
\min_{\mathbf K} ~ & h(\mathbf K)  \nonumber\\
\mbox{s.t. }~ & T_{\infty}(\mathbf K) < \gamma
\end{align}
\noindent where 
\begin{equation}
h(\mathbf K) \triangleq \left(J(\mathbf K) + \frac{b_k}{2} \vert| \mathbf K - \mathbf X^k \vert|_F^2 \right).
\end{equation}

\noindent We propose to solve (\ref{Kmin:eq}) using a feasible direction method in the search space of $\mathbf K$, summarized in Algorithm \ref{palm-Kmin:alg}. Starting from an interior point of the feasible region of the problem (In step 3.2 of Algorithm \ref{palm-alg1:alg}, $\mathbf K^k$ always satisfies $T_{\infty}(\mathbf K^k) < \gamma$), the algorithm first descends along the gradient of $h(\mathbf K)$ until the solution reaches a stationary point in the interior (line 6) or on the boundary of the constraint set. When the solution is in the interior of the feasible region, a gradient-descent update step is used (line 9). When the current solution is at the boundary and the gradient-descent direction violates the $H_{\infty}$-norm constraint, we seek a direction that reduces the minimization objective and simultaneously moves the solution away from the boundary of the $H_{\infty}$-norm constraint set to its interior (lines 12-13 in Algorithm \ref{palm-Kmin:alg}).

In lines 12-13, we find the improving feasible direction for (\ref{Kmin:eq}) when the solution is at the boundary of the feasible region. We recall the Zoutendijk's method \cite{bazaraa2013nonlinear} as the foundation for general feasible direction methods, which requires evaluating the gradients of both the objective and the constraint functions, 
i.e., $\nabla_{\mathbf K}h(\mathbf K)$ and $\nabla_{\mathbf K}T_{\infty}(\mathbf K)$.
 In the original formulation of Zoutendijk's method (S.I \cite{Lian:aa}, Appendix B.2 in \cite{Lian:2019aa}), the gradient of the constraint function is evaluated to form conditions for the improving feasible direction. However, due to the difficulty in evaluating the gradient of an $H_{\infty}$ norm, we utilize the concept of level sets as in \cite{saeki2006static}, as well as their LMI interpretation, to develop an alternative condition.

In each step of the Zoutendijk's method, a linear programming subproblem is solved to find the improving feasible direction. The inequality $\mathrm{trace}[(\nabla_{\mathbf K} h(\mathbf K))^T\cdot \Delta \mathbf K]<0$ guarantees that an update direction $\Delta \mathbf K$ decreases $h(\mathbf K)$ in (\ref{Kmin:eq}). Moreover, the inequality which involves gradient of the $H_{\infty}$ norm, i.e.,  $\mathrm{trace}[(\nabla_{\mathbf K} ||T_{z_1w_1}(\mathbf K)||_{\infty})^T\cdot \Delta \mathbf K] < 0$ can be used to check if $\Delta \mathbf K$ moves away from the $H_{\infty}$ bound. 
In the gain space of $\mathbf K\in \mathbb{R}^{m \times n}$, the set of all stabilizing $\mathbf K$ which satisfy (\ref{LMI_inf:eq}), i.e., with $H_{\infty}$ norm smaller than $\gamma$, is a level set
\begin{equation}
\label{palm-kappa:eq}
\mathcal{K}(\gamma):=\{ \mathbf K | ~||T_{z_1w_1}(\mathbf K)||_{\infty} <\gamma \}.
\end{equation}
Given a stabilizing gain $\mathbf K^0$, the algorithm in \cite{saeki2006static} proceeds by first finding a sufficiently small $\gamma_0$ such that $\mathbf K^0\in \mathcal{K}(\gamma_0)$. Next, a convex subset $\hat{\mathcal{K}}(\gamma_0)$ of $\mathcal{K}(\gamma_0)$, which also contains $\mathbf K^0$ near the boundary, can be formed using an \gls{LMI} sufficient condition. Then an inner point $\mathbf K^{\mathrm{in}}$ of $\hat{\mathcal{K}}(\gamma_0)$ is found using the following sufficient \gls{LMI} condition.

For the above ($\mathbf K^0$, $\gamma_0$),  $\mathbf K^{\mathrm{in}}$ is an inner point of $\hat{\mathcal{K}}(\gamma_0)$ if the matrix function $\mathbf G(\mathbf K^{\mathrm{in}}; \mathbf K^0)$ is positive definite, i.e.,
\begin{equation}
\mathbf G(\mathbf K^{\mathrm{in}}; \mathbf K^0) \succ 0 \Rightarrow \mathbf K^{\mathrm{in}} \in \hat{\mathcal{K}}(\gamma_0).
\label{LMIsuff:eq}
\end{equation}
The details of computing $\mathbf G(\mathbf K^{\mathrm{in}}; \mathbf K^0)$ are provided in \cite{Lian:2018aa}.

We combine the LMI condition (\ref{LMIsuff:eq}) and the gradient of $h(\mathbf K)$ in (\ref{Kmin:eq}) to form the iterative algorithm for solving (\ref{Kmin:eq}). The gradient of $h(\mathbf K)$ is
\begin{equation}
\nabla_{\mathbf K} h(\mathbf K) = 2(\mathbf R \mathbf K \mathbf C - \mathbf B^T \mathbf P)\mathbf L\mathbf C^T + b_k(\mathbf K - \mathbf X^k).
\end{equation}
Thus, given a current solution $\mathbf K^{\mathrm{cur}}$ near the boundary of the $H_{\infty}$-norm constraint set, an improving feasible point $\mathbf K^{\mathrm{in}}$ can be found by solving the following linear matrix inequality:
\begin{eqnarray}
\underset{z, \mathbf K^{\mathrm{in}}}{\mbox{max}} & & z \nonumber\\
\mbox{s.t.}& & \mathrm{trace}[(\nabla_{\mathbf K}h(\mathbf K^{\mathrm{cur}}))^T(\mathbf K^{\mathrm{in}}-\mathbf K^{\mathrm{cur}})] + z \leq 0 \nonumber\\
& &\mathbf G(\mathbf K^{\mathrm{in}}; \mathbf K^{\mathrm{cur}}) -\theta z\cdot\mathbf I \succeq 0.
\label{fd:eq}
\end{eqnarray} 
\noindent The parameter $\theta \geq 0$ is a predetermined factor that controls how far $\mathbf K$ moves away from the $H_{\infty}$-norm boundary. The value of $\theta$ determines the speed of reduction of the $H_{\infty}$ norm, with a small value of $\theta$ resulting in a less aggressive shrinkage of the $H_{\infty}$ norm. If the solution $z^*$ in (\ref{fd:eq}) is positive, then $\mathbf K^{in} - \mathbf K^{\mathrm{cur}}$ is an improving feasible direction; otherwise, an improving feasible direction cannot be found.

Given the current solution $\mathbf K^{\mathrm{cur}}$ and the inner point $\mathbf K^{\mathrm{in}}$ solved in (\ref{fd:eq}), the update rule is given in lines 12--13 of Algorithm \ref{palm-Kmin:alg},
\noindent where $d' \leq 1$ is the step size found by a backtracking line search using the Armijo condition \cite{bazaraa2013nonlinear}. 


\section{Sparsity-constrained noncooperative games for multi-agent control}

\label{palm-game:sec}

\subsection{Multi-agent model and generalized Nash equilibrium}

Next, we extend the optimization in Section \ref{palm-centra:sec} to the case when the agents have different optimization objectives. To accommodate this scenario, we consider the following multi-agent system with model uncertainty in $\mathbf A$ and $\mathbf B$ matrices. Consider a network of $N$ agents, where agent $i$ employs its control strategy $\mathbf u_i(t) \in \mathbb{R}^{q_i \times 1} $, $i=1,...,N$. Thus, (\ref{uncertain:eq}) becomes
\begin{align}
\label{multi-uncertain:eq}
&\dot{\mathbf x}(t) = (\mathbf A {+} \Delta \mathbf A) {\mathbf x}(t) + \sum_{i=1}^N{\left(\mathbf B_{(i)} {+} \Delta \mathbf B_{(i)}\right) \mathbf u_i(t)} + \mathbf B_2 \mathbf w_2(t) \nonumber\\
&\mathbf y(t) = \mathbf{Cx}(t) \nonumber\\
&\mathbf u_i(t) = - \mathbf K_i\mathbf y(t),~ i = 1,...,N,
\end{align}
\noindent where $\mathbf A \in \mathbb{R}^{n \times n}$, $\mathbf{B}_{(i)}\in \mathbb{R}^{n \times q_i}$ represent the nominal values of the state and control matrix, respectively, for the $i$-th control input. We assume all agents know $\mathbf A$ and $\mathbf B_{(i)}$ for $i=1,...,N$, and the uncertain matrices $\Delta \mathbf A\in \mathbb{R}^{n\times n}$ and $\Delta \mathbf B \triangleq [\Delta \mathbf B_{(1)}, \Delta \mathbf B_{(2)},...,\Delta \mathbf B_{(N)}]$ satisfy the norm-bounded assumption (\ref{deltaA:eq}), where $\Delta \mathbf{B}_{(i)}\in \mathbb{R}^{q_i\times n}$. 
Note that $\mathbf B_{(i)}$ is the column block of $\mathbf B$ in (\ref{uncertain:eq}), with $\sum_{i=1}^{N}{\mathbf B_{(i)}\mathbf u_i = \mathbf B\mathbf u}$, and $\mathbf K_i \in \mathbb{R}^{q_i \times p}$ is the row block of $\mathbf K$ in (\ref{p2}) associated with the rows corresponding to the control inputs for agent $i$. Thus, the multi-agent system can be expressed in the form (\ref{sigma:eq}--\ref{sigmak:eq}), with the first equation in (\ref{sigma:eq}) replaced by
\begin{equation}
\dot{\mathbf x}(t) = \mathbf{Ax}(t) + \sum_{i=1}^{N}{\mathbf B_{(i)}\mathbf u_i(t)} + \mathbf B_1\mathbf w_1(t) + \mathbf B_2\mathbf w_2(t). 
\label{multi-sigma:eq}
\end{equation}
\noindent We introduce the following notation. 
Let $\mathbf K_{-i}$ denote the set of strategies $j\neq i, j=1,...,N$. When agent $i$ chooses its strategy $\mathbf K_i$ in (\ref{multi-uncertain:eq}) given $\mathbf K_{-i}$, we refer to the resulting feedback gain matrix $\mathbf{K}$ as $\{\mathbf K_i; \mathbf K_{-i} \}$.

In the multi-agent system (\ref{multi-sigma:eq}), the single performance output $\mathbf z_2$ in (\ref{uncertain:eq}) is replaced by $N$ individual performance outputs of the agents $\mathbf z_{2, (i)}$. Assuming that each performance output $\mathbf z_{2,(i)} = \mathbf C_{2, (i)}\mathbf{x} + \mathbf{D}_{2,(i)}\mathbf u_i$ has a form that satisfies (\ref{palm-h2:eq}), the $H_2$-cost from $\mathbf w_2$ to agent $i$'s performance output can equivalently be defined as the individual LQR cost of agent $i$: 
\begin{eqnarray}
J_i(\mathbf K) &=& \int_{t=0}^{\infty}{\left[\mathbf x^T(t)\mathbf Q_i\mathbf x(t) + \mathbf u_i^T(t)\mathbf R_i\mathbf u_i(t)\right]dt} \nonumber\\
\mbox{s.t.} & & \mathbf w_1(t) = \mathbf 0, \mathbf w_2(t) = \boldsymbol \delta(t)
\label{Ji:eq}
\end{eqnarray}
\noindent where $\mathbf Q_i \in \mathbb{R}^{n\times n} \succeq 0$ and $\mathbf R_i \in \mathbb{R}^{q_i \times q_i} \succ 0$ are weight matrices for state and control input of agent $i$, respectively, and $\mathbf w_2(t)$ is an impulse disturbance. Similar to the centralized case, for stabilization of (\ref{multi-uncertain:eq}), the joint control strategy $\mathbf K$ needs to satisfy (\ref{LMI_inf:eq}).

In addition, we are interested in implementing a sparse controller subject to a global sparsity constraint.
In the following, we develop a noncooperative game where agent $i$ is modeled as a game player, with its strategy given by the control policy represented by $\mathbf K_i$. The joint strategies $\{\mathbf K_1, \mathbf K_2, ..., \mathbf K_N\}$ must guarantee stability of the uncertain system (\ref{multi-uncertain:eq}) with at most $s$ communication links in total. {Thus}, the set of admissible strategies $\{\mathbf K_1, \mathbf K_2,..., \mathbf K_N\}$ must satisfy \begin{eqnarray}
\vert| T_{z_1w_1}(\{\mathbf K_1, \mathbf K_2,...,\mathbf K_N \})\vert|_{\infty} \allowbreak < \gamma,\\ \mathrm{card}(\{\mathbf K_1, \mathbf K_2,...,\mathbf K_N \}) \leq s, 
\end{eqnarray}
and the set of feasible strategies for player $i$, given other players' strategies $\mathbf K_{-i}$, must satisfy
\begin{align}
\label{ith-constraint:eq}
\mathcal{G}_i(\mathbf K_{-i}) = \{ \mathbf K_i \vert  \mathrm{card}(\{\mathbf K_i; \mathbf K_{-i} \}) \leq s, \nonumber\\
\vert| T_{z_1w_1}(\{\mathbf K_i; \mathbf K_{-i} \})\vert|_{\infty} < \gamma.
\end{align}
\noindent Given $\mathbf K_{-i}$, player $i$ solves the following optimization:
\begin{align}
\label{ith-game:eq}
&\min_{\mathbf K_i} J_i(\{\mathbf K_i; \mathbf K_{-i}\}) \nonumber\\
&\mathrm{s.t. } \; \mathbf K_i \in \mathcal{G}_i(\mathbf K_{-i}).
\end{align}

\noindent Following \cite{paccagnan2016distributed}, we can say that the set of strategies $(\mathbf K_1^*, \mathbf K_2^*,...,\mathbf K_N^*)$ is a Generalized Nash Equilibrium (GNE) if
\begin{align}
\label{GNE:eq}
J_i(\{\mathbf K^*_i; \mathbf K^*_{-i}\}) \leq J_i(\{\mathbf K_i; \mathbf K^*_{-i}\}),  ~\forall \mathbf K_i \in \mathcal{G}_i(\mathbf K^*_{-i}), \nonumber\\
i = 1,...,N.
\end{align}
In \gls{GNE}, no user can unitarily deviate from the equilibrium to improve his utility given that the strategy satisfies the global constraint \cite{paccagnan2016distributed}. A \gls{GNE} differs from \gls{NE} due to the presence of global constraints.

\subsection{PALM algorithm for GNE}
We propose to solve the generalized Nash strategies (\ref{GNE:eq}) using the best-response dynamic (\ref{ith-game:eq}) where each player takes its turn to maximize its payoff based on other players' strategies. The steps are listed in Algorithm \ref{palm-game:alg}. Recall Algorithm \ref{palm-alg1:alg}, where the tuple $\mathbf K, \mathbf F$ was iteratively optimized to solve the penalized optimization (\ref{Pc}). Similarly, given $\mathbf K_{-i}$, player $i$'s optimization (\ref{ith-game:eq}) can be written in the penalized form using indicator functions as
\begin{equation}
\label{Phi_i:eq}
\min_{\mathbf K_i, \mathbf F} \Phi_i(\mathbf K_i, \mathbf F; \mathbf K_{-i})
\end{equation}
\noindent with
\begin{eqnarray}
\Phi_i(\mathbf K_i, \mathbf F; \mathbf K_{-i}) \triangleq &J_i(\{\mathbf K_i; \mathbf K_{-i}\}) + h(\{ \mathbf K_i; \mathbf K_{-i}\}) \nonumber\\
 &+ f(\mathbf F) + H(\{\mathbf K_i; \mathbf K_{-i}\}, \mathbf F),
\end{eqnarray}
\noindent where the indicator functions $h(\cdot)$ and $f(\cdot)$ are given by (\ref{palm-gK:eq},\ref{palm-fF:eq}), and the matrix $\{ \mathbf K_i; \mathbf K_{-i}\}$ is defined after (\ref{multi-sigma:eq}) . In this optimization, $\mathbf K_i \in \mathbb{R}^{q_i \times n}$ is viewed as the feedback gain of agent $i$ that satisfies $\vert| T_{z_1w_1}(\{\mathbf K_i; \mathbf K_{-i} \})\vert|_{\infty} < \gamma$, and $\mathbf F \in \mathrm{R}^{m\times n}$ represents the system-wide sparse feedback gain matrix that satisfies the global sparsity constraint.
In the function $\Phi(\mathbf K, \mathbf F)$ (\ref{Pc}), the variables $\mathbf K$ and $\mathbf F$ were of the same size, and they represented the same global sparsity-constrained feedback gain. However, when minimizing $\Phi_i(\mathbf K_i, \mathbf F; \mathbf K_{-i})$ (\ref{Phi_i:eq}), the variable $\mathbf K_i$ is the robust feedback gain for player $i$ while $\mathbf F$ represents the global feedback gain that satisfies the sparsity constraint. In the best-response dynamic, in each round the players take turns to minimize their own respective $\Phi_i$ functions over $\mathbf K_i$ and $\mathbf F$. The equilibrium point is achieved when no player can improve its $\Phi_i$ using $\mathbf K_i$ and $\mathbf F$ while $\mathbf K_j$ is fixed for $j\neq i$. 
Note that in the initial best response update steps, given non-sparse $\mathbf K_{-i}$, the minimization objective (\ref{Phi_i:eq}) cannot drive the coupling function $H(\{\mathbf K_i; \mathbf K_{-i}\}, \mathbf F)$ in (\ref{H:eq}) to zero. This is because $\sum_{i\neq j}{\vert| \mathbf K_i - (\mathbf F)_i \vert|^2_F}\approx 0$ only when $\mathbf K_{-i}$ approaches the desired level of sparsity, where $(\mathbf F)_i\in \mathbb{R}^{q_i \times n}$ denotes the row block of $\mathbf{F}$ that corresponds to the feedback gain of the $i^{th}$ player. 

\begin{algorithm}[!b]
     \caption{PALM algorithm for computing GNE (\ref{GNE:eq})}
   \label{palm-game:alg}
  \begin{algorithmic}[1]
  \State {\bf Given} $s$: global sparsity constraint, $\gamma$: $H_{\infty}$-norm bound.
  \State  {\bf Initialization}: 
  \State $\mathbf K^0$: any stabilizing feedback gain with $T_{\infty}(\mathbf K^0) < \gamma$.
  \State $\mathbf F^0$: any stabilizing feedback gain $\mathbf F^0$. 
  \For{$l = 1... l_{\max}$ until $\vert| \mathbf F^{1}-\mathbf F^{l-1}\vert|_F < \epsilon_3$}
  \State $\mathbf K^{l}:=\mathbf K^{l-1}$, $\mathbf F^{l}:=\mathbf F^{l-1}$
  \For{$i=1...N$}
  \State // Solve using (\ref{Phi_i-Z:eq}--\ref{Phi_i-Kmin:eq}) with $\mathbf K_i^{l}$, $\mathbf F^{l}$ as the initial values:
  \State $
  \hat{\mathbf K}_i, \hat{\mathbf F} = \arg\min_{ \mathbf K_i, \mathbf F} \Phi_i(\mathbf K_i, \mathbf F; \mathbf K^{l}_{-i})$
  \State // Update $\mathbf K^l$ and $\mathbf F^l$:
  \State $\mathbf K^{l} = \{ \hat{\mathbf K}_i; \mathbf K^{l}_{-i}\}$
  \State $\mathbf F^{l} = \hat{\mathbf F}$
  \EndFor
  \EndFor
  \State {\bf Output: $\mathbf K^{\mathrm{GNE}}(s):= \mathbf F^l$.}
   \end{algorithmic}
\end{algorithm}

The minimization of (\ref{Phi_i:eq}) is similar to the minimization of (\ref{Pc}). Thus, modified Algorithm \ref{palm-alg1:alg} is used in line 9 of Algorithm \ref{palm-game:alg} to solve (\ref{Phi_i:eq}). Given its $\mathbf K_i^l$, $\mathbf F^l$ at iteration $l$, the following proximal operators are performed by player $i$ in the minimization of line 9 of Algorithm \ref{palm-game:alg}.

\subsubsection{$\mathbf F$-minimization}:

\noindent Compute the proximal point $\mathbf Z^k$ for $\mathbf F^k$:
\begin{align}
\label{Phi_i-Z:eq}
\mathbf Z^k &= \mathbf F^k - \frac{1}{a} \nabla_{\mathbf F} H(\mathbf K^k, \mathbf F^k) \nonumber\\
&= \mathbf F^k - \frac{\gamma}{a}(\mathbf F^k - \{\mathbf K_i^k; \mathbf K_{-i}\})
\end{align}
\noindent Solve the proximal operator:
\begin{eqnarray}
&\mathbf F^{k+1} = &\arg\min_{\mathbf F} \frac{a}{2}\vert| \mathbf F - \mathbf Z^k \vert|_F^2 \nonumber\\
&~&\mathrm{s.t. }~ \mathrm{card}(\mathbf F) \leq s,
\end{eqnarray}
\noindent and get $\mathbf F^{k+1} = [\mathbf Z^k]_s$, similar to Step 2 of Algorithm \ref{palm-alg1:alg}.

\subsubsection{$\mathbf K$-minimization}:

\noindent Compute proximal point $\mathbf X_i^k$ for $\mathbf K_i$:
\begin{eqnarray}
\mathbf X_i^k &= &\mathbf K_i^k - \frac{1}{b} \nabla_{\mathbf K_i} H(\mathbf K_i^k - (\mathbf F^{k+1})_i) \nonumber\\
&=& \mathbf K_i^k - \frac{\rho}{b}(\mathbf K_i^k - (\mathbf F^{k+1})_i).
\end{eqnarray} 
\noindent Solve the proximal operator:
\begin{eqnarray}
\label{Phi_i-Kmin:eq}
\mathbf K_i^{k+1} &=& \arg\min_{\mathbf K_i} \left\{ J_i(\{\mathbf K_i; \mathbf K_{-i}\})+ \frac{b}{2}\vert| \mathbf K_i - \mathbf X_i^k \vert|_F^2 \right\}\nonumber\\
 & & \mathrm{s.t.}~\vert| T_{z_1w_1}(\{ \mathbf K_i; \mathbf K_{-i}\}) \vert|_{\infty} < \gamma.
\end{eqnarray}

Solving (\ref{Phi_i-Kmin:eq}) is similar to solving (\ref{Kmin:eq}). In (\ref{Phi_i-Kmin:eq}), player $i$ aims to update its control strategy $\mathbf K_i$ given other players' strategies $\mathbf K_{-i}$. Algorithm \ref{palm-Kmin:alg} is applied to solve (\ref{Phi_i-Kmin:eq}) with several modifications. The minimization cost in (\ref{Phi_i-Kmin:eq}) is defined as $h_i(\mathbf K_i, \mathbf K_{-i})\triangleq J_i(\{\mathbf K_i; \mathbf K_{-i}\})+ \frac{b}{2}\vert| \mathbf K_i - \mathbf X_i^k \vert|_F^2$, and the gradient with respect to $\mathbf K_i$ in line 4 of Algorithm \ref{palm-Kmin:alg} is replaced by 
\begin{align}
&\nabla_{\mathbf K_i}h_i(\mathbf K_i, \mathbf K_{-i}) = \\
&2(\mathbf R_i\mathbf K_i\mathbf C - \mathbf B_{(i)}^T\mathbf P(\mathbf K_i, \mathbf K_{-i}))\mathbf L(\mathbf K_i, \mathbf K_{-i})\mathbf C^T + b(\mathbf K_i {-} \mathbf X^k_{i}) \nonumber
\end{align}
\noindent where $L(\mathbf K_i, \mathbf K_{-i})$ and $P(\mathbf K_i, \mathbf K_{-i})$ are the solution of the following set of equations:
\begin{align}
\label{palm-lyapPi:eq}
&(\bar{\mathbf A}_{cl}(\mathbf K_i, \mathbf K_{-i}))^{\mathrm{T}}\boldsymbol P(\mathbf K_i, \mathbf K_{-i}) + \boldsymbol P(\mathbf K_i, \mathbf K_{-i})\bar{\mathbf A}_{cl}(\mathbf K_i, \mathbf K_{-i})
 \nonumber\\
& + \bar{\mathbf Q}_i(\mathbf K_i, \mathbf K_{-i})= 0 \nonumber\\
&\bar{\mathbf A}_{cl}(\mathbf K_i, \mathbf K_{-i})\boldsymbol L(\mathbf K_i, \mathbf K_{-i}) + \boldsymbol L(\mathbf K_i, \mathbf K_{-i})(\bar{\mathbf A}_{cl}(\mathbf K_i, \mathbf K_{-i}))^{T} \nonumber\\
&+ \boldsymbol B_2 \boldsymbol B_2^{\mathrm{T}} = 0 
\end{align}
\noindent and
\begin{align}
&\bar{\mathbf A}_{cl}(\mathbf K_i, \mathbf K_{-i}) = \mathbf A - \sum_{j\neq i}{\mathbf B_{(j)}}\mathbf K_j \mathbf C - \boldsymbol B_{(i)}\boldsymbol K_i\mathbf C\\
&\bar{\mathbf Q}_i(\mathbf K_i, \mathbf K_{-i}) = \mathbf Q_i + \mathbf C^T(\sum_{j\neq i}{(\mathbf{K}_j)^{T}\mathbf R_j \mathbf K_j} +\mathbf{K}_i^{T}\boldsymbol R_i \boldsymbol K_i)\mathbf C.\nonumber
\end{align}
\noindent $\mathbf Q_i$ and $\mathbf R_i$ are defined in (\ref{Ji:eq}).
\noindent Similar to lines 10--14 in Algorithm \ref{palm-Kmin:alg}, when player $i$'s strategy $\mathbf K_i^{\mathrm{cur}}$ is near the boundary of the $H_{\infty}$-norm constraint given other players' strategies $\mathbf K_{-i}$, an improving feasible direction for $\mathbf K_i$ can be found by solving an \gls{LMI} such as (\ref{fd:eq}) for scalar $z$ and $\mathbf K_i^{in}\in \mathbb{R}^{q_i\times n}$:
\begin{align}
&\underset{z, \mathbf K_i^{\mathrm{in}}}{\mbox{Maximize}} ~  z \nonumber\\
&\mbox{s.t.}  ~\mathrm{trace}[(\nabla_{\mathbf K_i}h_i(\mathbf K_i^{\mathrm{cur}}, \mathbf K_{-i})^T(\mathbf K_i^{\mathrm{in}}-\mathbf K_i^{\mathrm{cur}})] + z \leq 0 \nonumber\\
& \qquad {\mbox{    }} \mathbf G\left(
\{\mathbf K_i^{\mathrm{in}}; \mathbf K_{-i}\}
; 
\{\mathbf K_i^{\mathrm{cur}}; \mathbf K_{-i}\}
\right) -\theta z\cdot\mathbf I \succeq 0.
\label{fd_oneplayer:eq}
\end{align}

\noindent Here, $\theta$ is the factor to control the speed of reduction of $H_{\infty}$ norm, and $G$ is defined as in (\ref{LMIsuff:eq}). Then the update direction of $\mathbf K_i$ can be formed as $\Delta \mathbf K_i = \mathbf K_i^{\mathrm{in}} - \mathbf K_i^{\mathrm{cur}}$. We note that in $\mathbf K$-minimization step, each player updates its own strategy $\mathbf K_i$, while in $\mathbf F$-minimization step, the strategies of all the players are jointly updated. Thus, Algorithm \ref{palm-game:alg} has partially distributed computation.

Finally, we note that the centralized problem (\ref{P}) can be represented as a potential game by modifying the noncooperative game (\ref{GNE:eq}). A game $\{N, \{\mathcal{A}_i\}, \{J_i\} \}$ with $N$ players, action set $\{\mathcal A_i\}_{i=1}^N$ and utilities $\{J_i\}_{i=1}^N$, is an exact potential game \cite{li2013designing} if there exists a global function $\Phi$, such that for every player $i\in N$, $a_{-i} \in \mathcal A_{-i}$ and $a_i', a_i^{''}\in \mathcal A_i$, 
\begin{equation}
\label{palm-potential:eq}
J_i(a_i', a_{-i}) - J_i(a_i^{''}, a_{-i}) = \Phi(a_i', a_{-i}) - \Phi(a_i^{''}, a_{-i}).
\end{equation}
We employ a common assumption that the input penalty of each user is uncorrelated, i.e., $\mathbf u^T \mathbf R \mathbf u = \sum_{i=1}^N{\mathbf u_i^T \mathrm R_{(i)} \mathbf u_i}$, where $\mathbf R_{(i)}$ is the submatrix of $\mathbf R$ that represents the weight matrix for $\mathbf u_i(t)$.  Thus, the objective $J(\mathbf K)$ in (\ref{PALM-J-centra:eq}) can be expressed as
\begin{align}
\label{palm-J-potential:eq}
&J(\mathbf K)=J(\{\mathbf K_i, \mathbf K_{-i}\}) = \\
&\int_{0}^{\infty}{\left[\mathbf x^T\left(\mathbf Q + \mathbf C^T(\sum_{j\neq i}{\mathbf K_j^T \mathbf R_{(j)} \mathbf K_j})\mathbf C\right)\mathbf x {+} \mathbf u_i^T \mathbf R_{(i)} \mathbf u_i \right]dt}, \nonumber
\end{align}
\noindent
with $\mathbf u_i= -\mathbf K_i\mathbf y$. Thus, the minimization objectives $J_i(\{\mathbf K_i; \mathbf K_{-i}\})$ of all players in (\ref{ith-game:eq}) are replaced by the global LQR cost (\ref{palm-J-potential:eq}).
To convert the game in (\ref{ith-game:eq}) into an exact potential game, we set $\mathbf Q_i = \mathbf Q + \mathbf C^T(\sum_{j\neq i}{\mathbf K_j^T \mathbf R_{(j)} \mathbf K_j})\mathbf C$, $\mathbf R_i = \mathbf R_{(i)}$ in the individual cost (\ref{Ji:eq}), which is consistent with (\ref{palm-J-potential:eq}). The \gls{GNE} strategies $\mathbf K^*_1, \mathbf K^*_2, ..., \mathbf K^*_N$ for the potential game can be written as
\begin{align}
&J(\{\mathbf K^*_i; \mathbf K^*_{-i}\}) \leq J(\{\mathbf K_i; \mathbf K^*_{-i}\}),  \nonumber\\
&~\forall \mathbf K_i \in \mathcal{G}_i(\mathbf K^*_{-i}), \;i = 1,...,N.
\label{palm-pg:eq}
\end{align} 
\noindent We employ Algorithm \ref{palm-game:alg} to compute (\ref{palm-pg:eq}). Players update their control strategies in the $\mathbf K$-minimization step distributively, and jointly update their strategies in the $\mathbf F$-minimization step, thereby obtaining a partially distributed implementation of the centralized sparsity-constrained problem (\ref{P}).

\begin{figure}[t]
\centering
\includegraphics[width=\columnwidth]{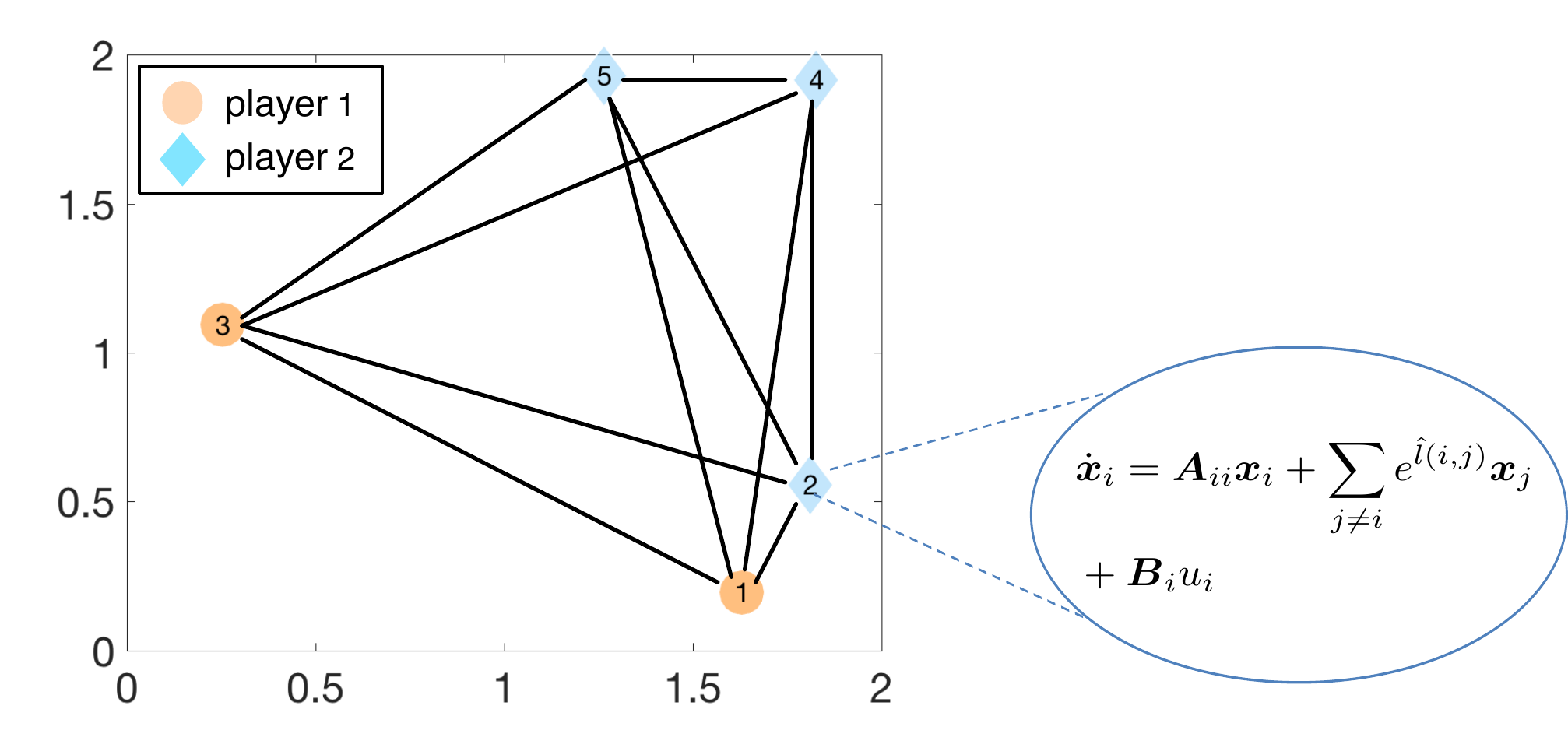}
\caption{The 5-node open-loop unstable network}
\label{palm-5nodesys:fig}
\end{figure}

\section{Numerical results and convergence analysis}
\label{palm-num:sec}

\subsection{Network model}

We consider an example of an uncertain network model from \cite{motee2008optimal}. The network consists of $N$ connected nodes distributed randomly on a $L$ unit by $L$ unit square area. Each node is an unstable second-order system coupled with other nodes through an exponentially decaying function of the Euclidean distance $\hat l(i,j)$ \cite{mihailo, motee2008optimal}. The state-space representation of node $i$ is given as:
\begin{align}
\begin{bmatrix}\dot{x}_{1i} \\ \dot{x}_{2i}\end{bmatrix} {=} \hat{
\mathbf A}_{ii} \begin{bmatrix} {x}_{1i} \\ {x}_{2i}\end{bmatrix} 
{+} \sum_{j\neq i}{e^{-\hat l(i,j)}} \begin{bmatrix}{x}_{1j} \\ {x}_{2j}\end{bmatrix} 
{+} \begin{bmatrix}0\\1\end{bmatrix}(d_i+u_i).
\label{unst:eq}
\end{align}
\noindent In the above state-space representation, the state matrix $\hat{\mathbf A}_{ii}, ~i=1,...,N$ and the Euclidean distances $\hat l(i,j)$ are not known exactly. In particular, 
\begin{eqnarray}
\hat{\mathbf A}_{ii} &= &\mathbf A_{ii} + \mathbf A_{ii}\odot \begin{bmatrix} \theta_{11} & \theta_{12} \\
\theta_{21} & \theta_{22}\end{bmatrix}\nonumber\\
\hat l(i,j) &=& l(i,j) \cdot (1+\delta_{i,j}),
\end{eqnarray}
\noindent where $\mathbf A_{ii}$ and $l(i,j)$ are the nominal values, and $\delta_{ij}$ and $\theta_{ij}$ are independent random perturbations, uniformly distributed in the range $\pm 20\%$. The operator $\odot$ denotes element-wise multiplication. As in (\ref{uncertain:eq}), $\mathbf A$ denotes the nominal value of the state matrix of this $N$-node unstable system, and $\hat{\mathbf A}$ denotes one realization of the perturbed state matrix. The uncertain matrix $\Delta \mathbf A = \hat{\mathbf A} - \mathbf A$ in (\ref{uncertain:eq}). The control input matrix is assumed to be known for this example, so that 
$\hat{\mathbf B}= \mathbf B = \mathbf{1}_N \otimes \mathbf B_{ii}$, where $\mathbf B_{ii}=\begin{bmatrix}0 & 1\end{bmatrix}^T$, and $\otimes$ denotes the Kronecker product \cite{meyer2000matrix}.
In this simulation study, we collected 200 random samples of $\hat{\mathbf A}$. To guarantee closed-loop stability of (\ref{unst:eq}), we numerically compute the worst-case $\hat{\boldsymbol A}$ as,
\begin{equation}
\hat{\boldsymbol A}_{\mathrm{worst}} = \arg\max_{\hat{\boldsymbol A}}\sigma_{\mathrm{max}}(\hat{\boldsymbol A}-\boldsymbol A).
\end{equation}
Using the singular value decomposition, we obtain $\boldsymbol U\boldsymbol S \boldsymbol V^T {=} \hat{\boldsymbol A}_{\mathrm{worst}}{-}\boldsymbol A$.
Normalizing $\boldsymbol S$ by $\sigma_{\mathrm{max}}(\boldsymbol S)$, we set $\boldsymbol B_1 = \sqrt{\sigma_{\mathrm{max}}(\boldsymbol S)}\boldsymbol U$, $\boldsymbol C_1=\sqrt{\sigma_{\mathrm{max}}(\boldsymbol S)}\boldsymbol V^T$ in (\ref{deltaA:eq}).
Due to this normalization, $\gamma=1$.

The following parameters are employed in the simulations. We set $L=2$ and $N=5$, thus $\mathbf A\in \mathbb{R}^{10 \times 10}$, $\mathbf B \in \mathbb{R}^{10 \times 5}$. The output matrix $\mathbf C = \mathbf I_{10}$. The dense feedback matrix $\mathbf K$ has $\mathrm{card}(\mathbf K)=50$. When the feedback controller is completely decentralized, i.e., feedback links only exist between states and controllers within the same node, $\mathrm{card}(\mathbf K)=10$. The performance index for the LQR cost employs $\mathbf Q=100\cdot\mathbf I$ and $\mathbf R=\mathbf I$ in (\ref{qr:eq}) for the centralized problem (\ref{P}). For the noncooperative game (\ref{GNE:eq}), we consider a two-player game as shown in Figure \ref{palm-5nodesys:fig}, where player 1 is in charge of the control inputs in nodes 1 and 3 and player 2 is in charge of the control inputs in nodes 2, 4, 5. The performance index matrices $\mathbf Q_i, \mathbf R_i$, $i=1,2$ for the LQR cost in (\ref{Ji:eq}) satisfy:
\begin{align}
\label{palm-selfishQi:eq}
&\mathbf x^T \mathbf Q_1 \mathbf x {+} \mathbf u_1^T\mathbf R_1 \mathbf u_1 = 100 [(x_{11} {-} x_{13})^2 {+} (x_{21} {-} x_{23})^2] {+} u_1^2 {+} u_3^2   \nonumber\\
&\mathbf x^T \mathbf Q_2 \mathbf x + \mathbf u_2^T\mathbf R_2 \mathbf u_2 = 100\sum_{j=2,4,5}{(x_{1j}^2 + x_{2j}^2)} + \sum_{j=2,4,5} u_j^2.
\end{align}
\noindent We solve all the \glspl{LMI} using the CVX package \cite{cvx}.

\subsection{Social optimization}

First, we present simulation results for the problem (\ref{P}) applied to the system in (\ref{unst:eq}) with $\gamma=1$ in (\ref{P}) over a range of $s$-values. We implement Algorithm \ref{palm-alg1:alg}, with the resulting feedback matrix  denoted as $\mathbf K_{\mathrm{palm}}^*(s)$. For the same problem (\ref{P}), we also use Algorithm \ref{palm-game:alg} applied to the potential game (\ref{palm-pg:eq}), with the solution denoted by $\mathbf K^*_{\mathrm{PALMPG}}(s)$, given the sparsity constraint $s$. For comparison, we also run the GraSP algorithm that was used in \cite{Lian:2018aa}, with the resulting feedback denoted by $\mathbf K_{\mathrm{GraSP}}^*(s)$, initialized by a stabilizing decentralized controller $\mathbf K_{\mathrm{dec}}$ with $\mathrm{card}({\mathbf K}_{\mathrm{dec}})=10$.
In general, GraSP needs to be initialized by a $\mathbf K_0$ that satisfies $\mathrm{card}(\mathbf K_0) \leq s$ and $T_{\infty}(\mathbf K_0) < \gamma$, which in reality might be difficult to find.
In contrast, the \gls{PALM}-based Algorithm \ref{palm-alg1:alg} of this paper does not rely on any such sparse initialization. Finally, we show performance of the dense mixed $H_2/H_{\infty}$ controller using the simple gradient method in \cite{kami2008gradient}.


\begin{figure}[h]
\centering
\subfloat[$J$ vs. sparsity constraint $s$.]{\includegraphics[width=0.5\columnwidth]{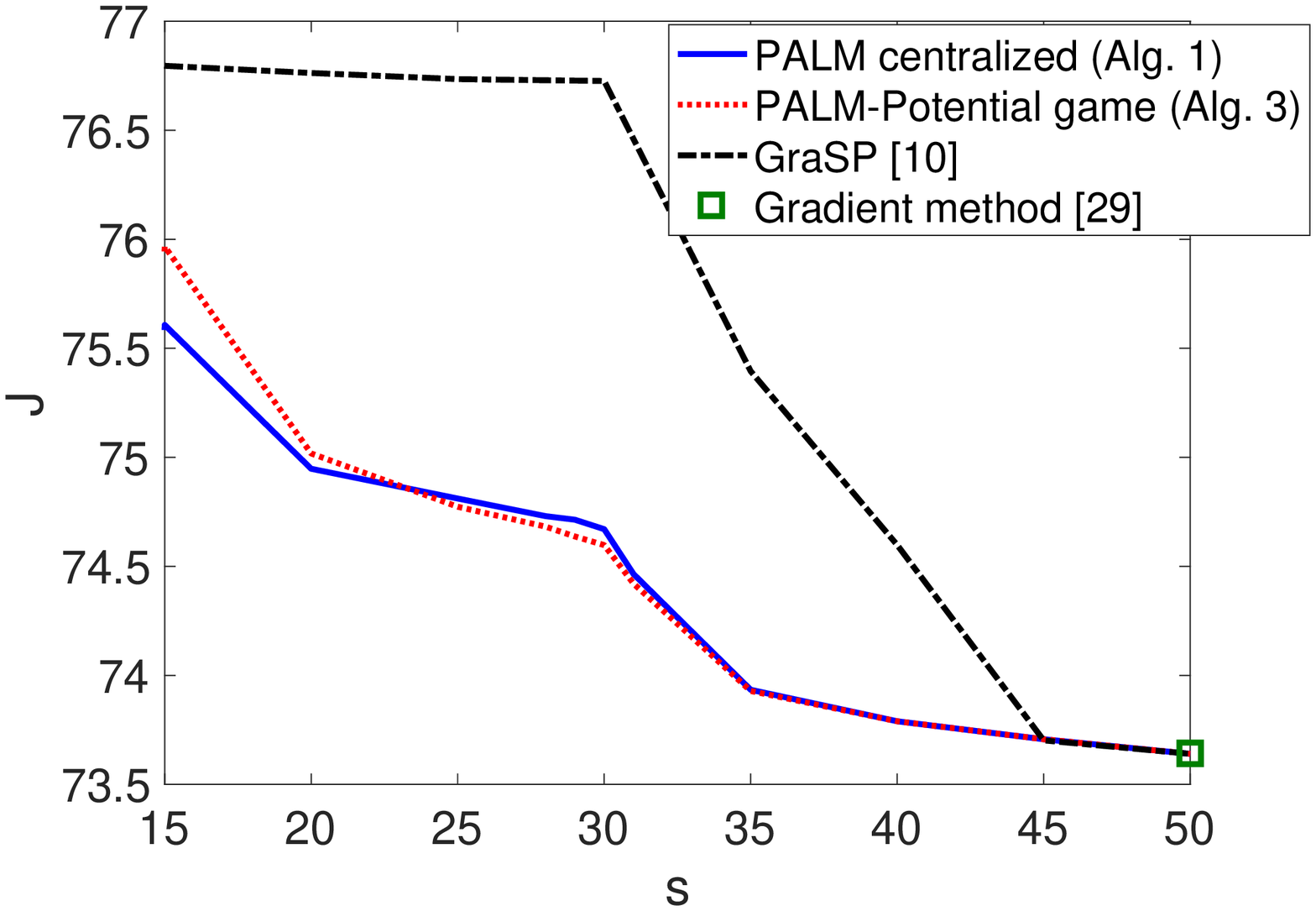}}
\subfloat[$H_{\infty}$ norm v.s. sparsity constraint $s$.]{
\includegraphics[width=0.5\columnwidth]{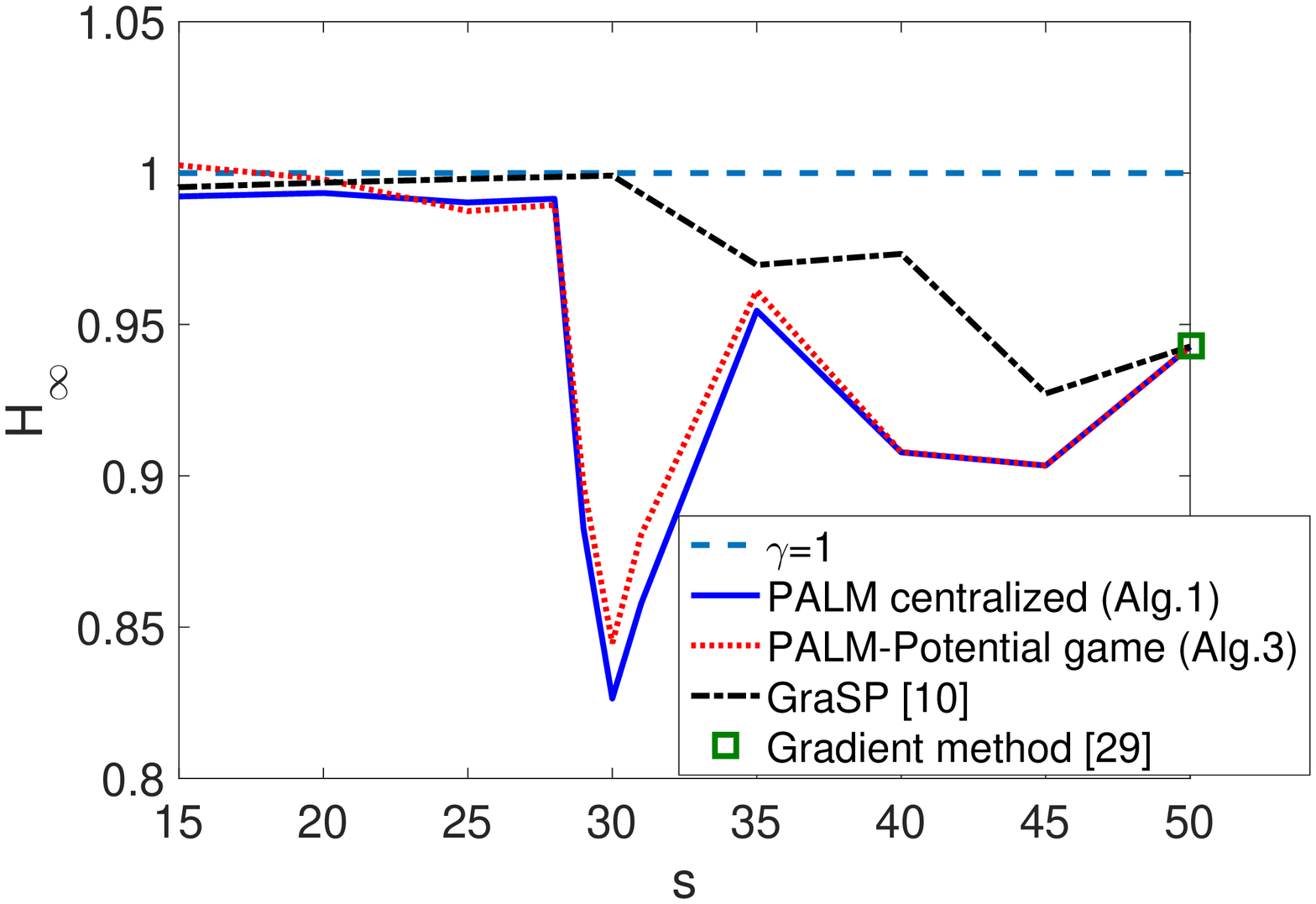}}
\caption{The LQR cost $J$ and $H_{\infty}$ norm vs. sparsity constraint $s$.
}
\label{PALM-J:fig}
\end{figure}

Figure \ref{PALM-J:fig} illustrates the optimal LQR cost $J$ in problem (\ref{P}) and the associated $H_{\infty}$ norm vs. sparsity constraint $s$.
For $15 \leq s \leq 50$, the centralized Algorithm \ref{palm-alg1:alg} and the potential game using Algorithm \ref{palm-game:alg} both converge to a solution with sufficiently small coupling function in (\ref{Pc}), which indicates $\mathbf F \approx \mathbf K$. From Figure \ref{PALM-J:fig}(a), we observe that the $H_2$ norms of all sparsity-constrained methods decrease as $s$ is relaxed, and approach to that of the dense controller \cite{dorfler2014sparsity}. However, the \gls{PALM}-based methods have similar LQR costs and outperform significantly the greedy \gls{GraSP} algorithm in \cite{Lian:2018aa}. In GraSP, the choice of active coordinates only depends on the gradient information of the function $J$. At convergence, the solution of the mixed $H_2/H_{\infty}$ problem has the sparsity structure given by the greedy selection step. For the PALM algorithm, since we iteratively compute the proximal map on $\mathbf X^k$ and $\mathbf Z^k$, the support is chosen based on the information on both the LQR cost $J(\mathbf K)$ and the $H_{\infty}$-norm constraint $T_{\infty}(\mathbf K)$. Thus, at convergence, the \gls{PALM} method finds a critical point of problem (\ref{P}) while \gls{GraSP} does not necessarily achieve it.
Figure \ref{PALM-J:fig}(b) shows the $H_{\infty}$ norms of $\mathbf K^*_{\mathrm{PALM}}(s)$, $\mathbf K^*_{\mathrm{PALMPG}}(s)$ and $\mathbf K^*_{\mathrm{GraSP}}(s)$. We observe that for both GraSP and PALM methods, the solution is found in the interior of the $H_{\infty}$-norm constraint for $s\geq 30$, and on the boundary for $s\leq 25$, which indicates that when the sparsity constraint becomes more stringent, satisfying the sparsity and $H_{\infty}$-norm constraints simultaneously becomes challenging.

Both Algorithm \ref{palm-alg1:alg} (the social optimization) and Algorithm \ref{palm-game:alg} (the potential game) are found to converge for all $s$-values for this system. 
Figure \ref{palm3_deltaK:fig} shows the error in consecutive steps for variable $\mathbf K$ at the end of step 3 of Algorithm \ref{palm-alg1:alg} as a function of iteration step, for different $s$-values. We found that $\Delta \mathbf F_k$ has a similar trend to $\Delta \mathbf K_k$. The errors in consecutive steps are defined as $\Delta \mathbf K^k \triangleq \mathbf K^k - \mathbf K^{k-1}$ and $\Delta \mathbf F^k \triangleq \mathbf F^k - \mathbf F^{k-1}$. We note that the error converges faster for larger $s$-values, which might be explained by the fact that that for $s > 25$, the minima are found in the interior of the $H_{\infty}$-norm constraint set (see Figure \ref{PALM-J:fig}).
 For Algorithm \ref{palm-game:alg} (potential game), the penalized cost function $\Phi_i$ and $\vert|\mathbf K - \mathbf F\vert|_F^2$ (line 9) have similar trends to those for Algorithm \ref{palm-alg1:alg}.
Moreover, it is demonstrated in Fig \ref{palm3_Phi:fig} that although Algorithm \ref{palm-alg1:alg} converges to a critical point of $\mathbf \Phi(\mathbf K, \mathbf F)$, the coupling function $H(\mathbf K, \mathbf F)>0$ for $s<15$. As a result, when Algorithm \ref{palm-alg1:alg} converges for these $s$-values, $\mathbf K \neq \mathbf F$, so a sparse feedback solution that satisfies (\ref{P}) cannot be found. Thus, in Figure \ref{PALM-J:fig}, we only show the LQR cost and $H_{\infty}$-norm for $15\leq s\leq 50$.

\begin{figure}[!t]
\centering
\includegraphics[width=1\columnwidth]{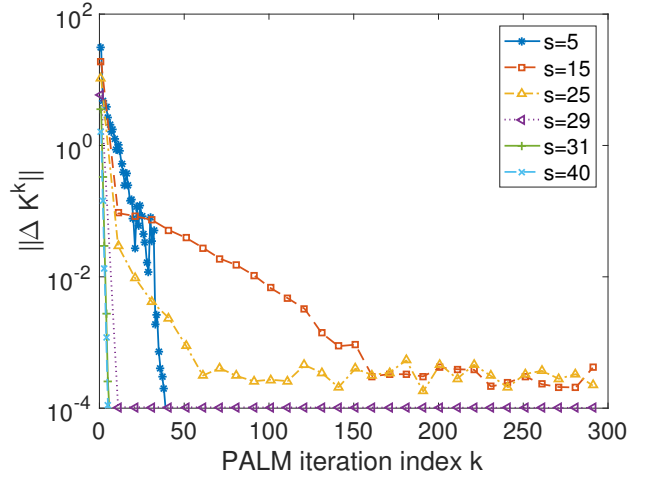}
\caption{The error in $\mathbf K$ vs. iteration $k$ in PALM Algorithm \ref{palm-alg1:alg} Step 2 and 3 for different $s$-values.}
\label{palm3_deltaK:fig}
\end{figure}

\begin{figure}[!t]
\centering
\includegraphics[width=\columnwidth]{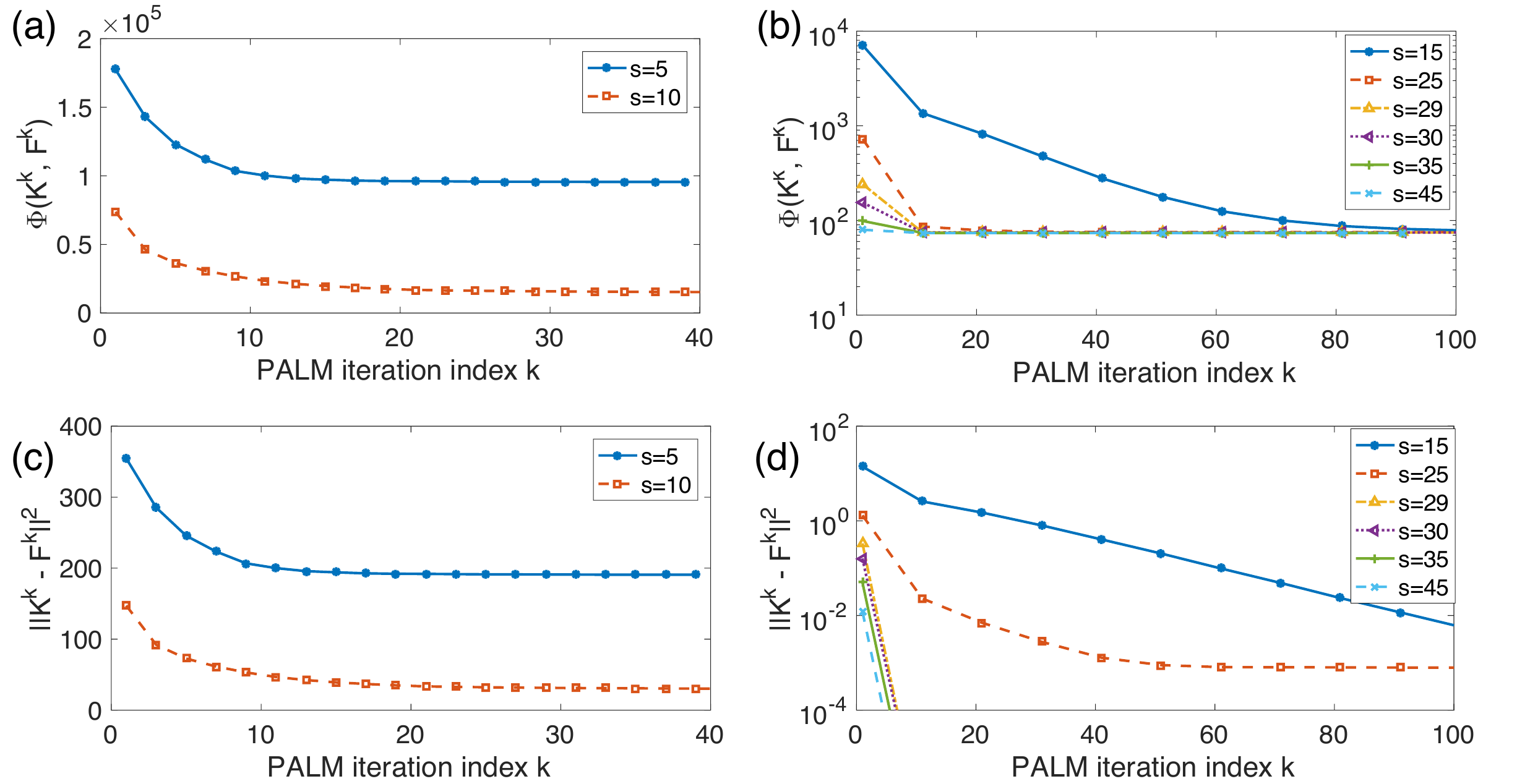}
\caption{The penalized cost function $\Phi(\mathbf K^k, \mathbf F^k)$ and the coupling function $\vert|\mathbf K^k - \mathbf F^k \vert|_F^2$ vs iteration $k$ in the end of Step 3 of Algorithm \ref{palm-alg1:alg} for multiple $s$-values.}
\label{palm3_Phi:fig}
\end{figure}


\subsection{The noncooperative game}

We investigate performance of Algorithm \ref{palm-game:alg} for the noncooperative game with different individual costs (\ref{palm-selfishQi:eq}) for the system (\ref{unst:eq}). We use $\mathbf K^{\mathrm{GNE}}(s) = \{\mathbf K_1^{\mathrm{GNE}}(s), \mathbf K_2^{\mathrm{GNE}}(s)\}$ to denote the two players' feedback produced by Algorithm \ref{palm-game:alg} when the sparsity constraint is given by $s$. 
Figure \ref{PALM-gne_delta:fig} shows the errors in consecutive steps of player $i$'s strategic variables $\mathbf K_i, \mathbf F_i$ for $i=1,2$ vs iteration round $l$ in Algorithm \ref{palm-game:alg}. We observe that both $\vert|\Delta \mathbf K_i \vert|_F$ and $\vert| \Delta \mathbf F_i\vert|_F$ decrease significantly within the first $10$ iterations and then saturate to small values as $l$ grows, resulting in the saturation of the penalized cost function $\Phi_i$ in line 9 of Algorithm \ref{palm-game:alg}, which corresponds to an approximate equilibrium point as discussed in section \ref{palm-game-convg:sec}. The normalized coupling function $\frac{1}{\rho}H(\mathbf K^l, \mathbf F^l)=\vert| \mathbf K^l - \mathbf F^l \vert|_F^2$ (\ref{H:eq}) decreases with iteration $l$, following the trend in Figure \ref{palm3_Phi:fig}.
For $s>20$, the square error $\vert| \mathbf K^l - \mathbf F^l \vert|_F^2$ reaches a sufficiently small value ($<10^{-4}$) at the equilibrium point, while for $s\leq 20$, the square error is larger, causing $T_{\infty}(\mathbf K^{\mathrm{GNE}}(s)) > T_{\infty}(\mathbf K^l)$, which results in $T_{\infty}(\mathbf K^{\mathrm{GNE}}(s))>1$ during  convergence.
One way to avoid this discrepency and still guarantee closed-loop stability is to replace $\gamma$ in (\ref{palm-gK:eq}) with $\gamma-\epsilon$, and provide a margin that compensates for the square error between $\mathbf K$ and $\mathbf F$. For this example, we set $\epsilon=0.01$, so that $T_{\infty}(\mathbf K^{\mathrm{GNE}}(s)) < 0.99$.

Figure \ref{PALM-gne_Ji:fig} illustrates the individual LQR costs $J_i$ as in (\ref{Ji:eq}), and the global $H_{\infty}$ norm when $\mathbf K^{\mathrm{GNE}}(s)$ is implemented. We observe that in Figure \ref{PALM-gne_Ji:fig}(a), for each player $i$, the LQR cost achieved at the equilibrium point $J_i(\mathbf K^{\mathrm{GNE}}(s))$ tends to decrease with $s$, which indicates that there is a trade-off between the selfish LQR cost and the global shared sparsity constraint. Figure \ref{PALM-gne_Ji:fig}(b) shows that  $T_{\infty}(\mathbf K^{\mathrm{GNE}}(s))<1$ for $15 \leq s \leq 45$, indicating that the Nash strategies in $\mathbf K^{\mathrm{GNE}}(s)$ are guaranteed to stabilize the uncertain system (\ref{unst:eq}). 
\begin{figure}[!t]
\centering
\includegraphics[width=\columnwidth]{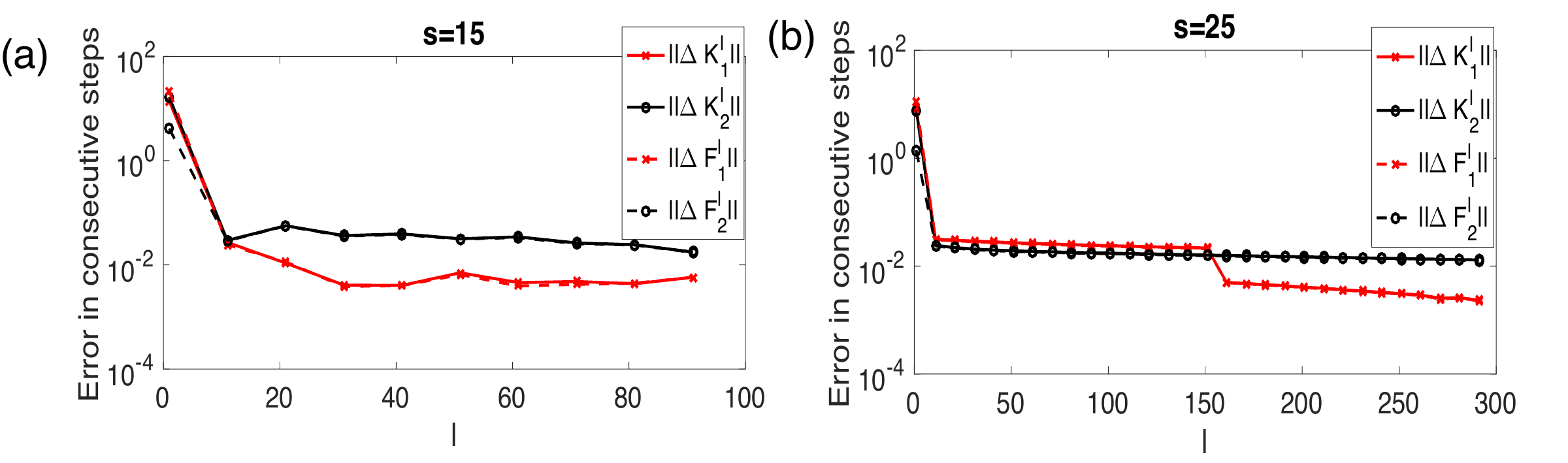}
\caption{Errors in consecutive steps of $\mathbf K_i^l$ and $\mathbf F_i^l $ for players $i=1,2$ vs. step $l$ in Algorithm \ref{palm-game:alg} (the noncooperative game).}
\label{PALM-gne_delta:fig}
\end{figure}

\begin{figure}[!t]
\centering
\subfloat[$J_i(\mathbf K^{\mathrm{GNE}}(s))$ vs. sparsity constraint $s$ at \gls{GNE} for $i=1,2$.]{
\includegraphics[width=1\columnwidth]{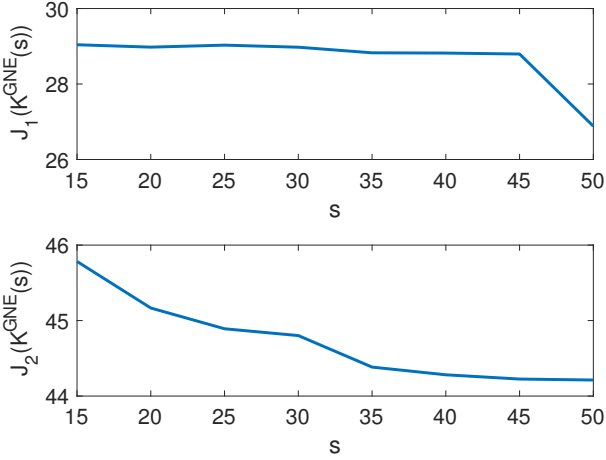}}\\
\subfloat[$T_{\infty}(\mathbf K^{\mathrm{GNE}}(s))$ vs. the sparsity constraint $s$ at \gls{GNE}.]{
\includegraphics[width=1\columnwidth]{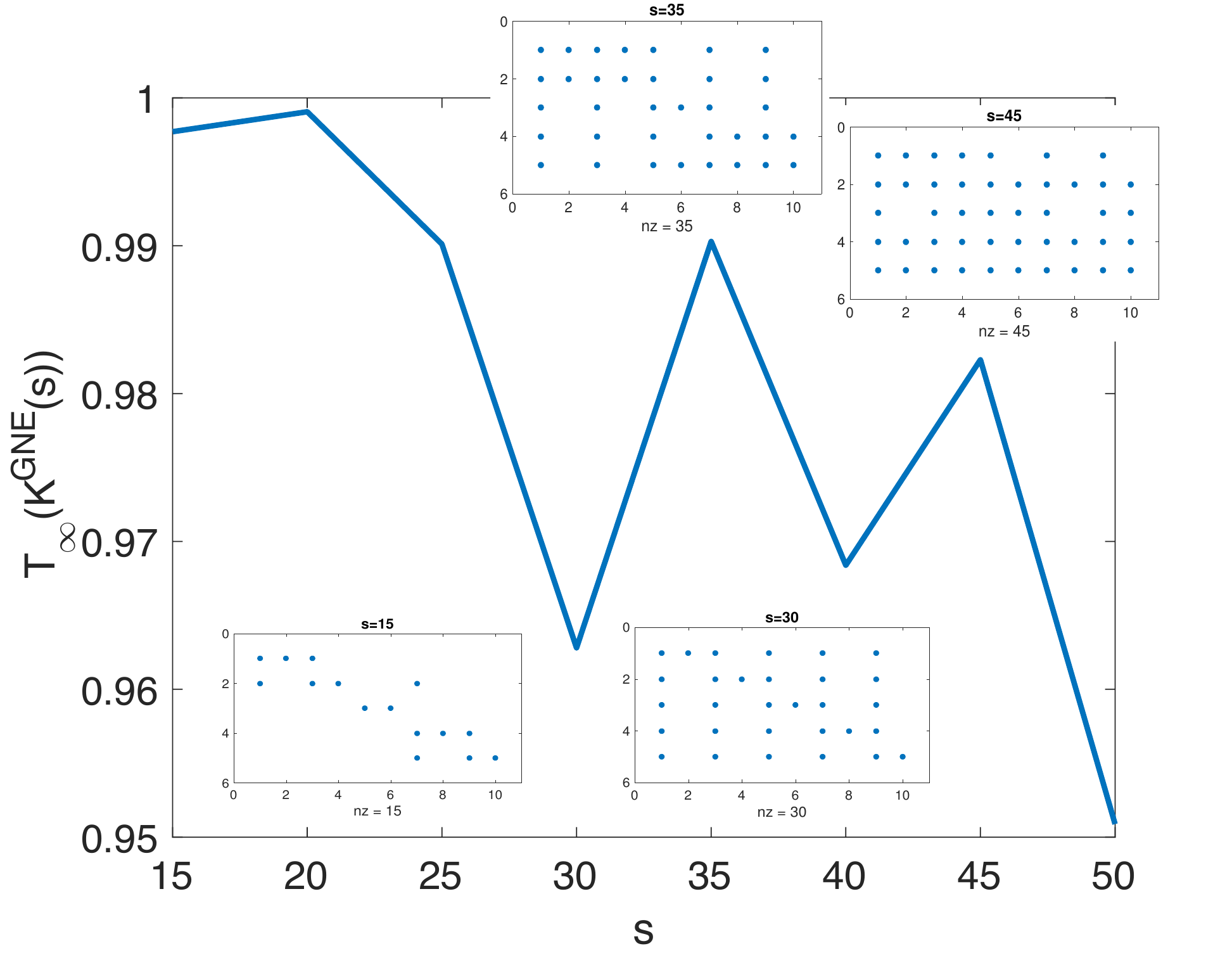}}
\caption{The individual LQR cost and global $H_{\infty}$ norm of $\mathbf K^{\mathrm{GNE}}(s)$ vs sparsity constraint $s$ at \gls{GNE} and the sparsity pattern of $\mathbf K^{\mathrm{GNE}}(s)$ for different $s$ values.}
\label{PALM-gne_Ji:fig}
\end{figure}

\subsection{Algorithm Convergence and Complexity}
\label{palm-game-convg:sec}
We close this section by providing some final comments about the convergence properties and complexity of the proposed algorithms. 
Global convergence of PALM for nonconvex nonsmooth functions was  studied in \cite{bolte2014proximal}, while that of PALM-based output feedback co-design under block-sparsity constraints was established in \cite{lin2017co}. Results in \cite{lin2017co,bolte2014proximal} are extended to analyze the convergence properties of Algorithm \ref{palm-alg1:alg} in \cite{Lian:aa,Lian:2019aa}. It has been proven in \citelatex{bolte2014proximalsupp} that if Lemma \ref{lowsemi:lem}--\ref{lip:lem} in \cite{Lian:aa} hold, then
the sequence generated by PALM algorithm globally converges. In addition, if Lemma \ref{KL:lem} of \cite{Lian:aa} holds, the sequence converges to a critical point \citelatex{bolte2014proximalsupp} of $\Phi$. This confirms convergence of Algorithm \ref{palm-alg1:alg} to a sparsity-constrained mixed $H_2/H_{\infty}$ controller, which corresponds to a critical point of $\Phi$ under mild assumptions on the functions $J$ and $g$.

Next, we briefly discuss the convergence properties of Algorithm \ref{palm-game:alg}. Suppose a \gls{GNE} (\ref{GNE:eq}) is given by $(\mathbf K^*_1, ..., \mathbf K^*_N)$. Then the following condition holds for each player $i$ \cite{hazan2017efficient}:
\begin{equation}
\nabla_{\mathcal{G}_i(\mathbf K^*_{-i}), \eta} J_i(\{ \mathbf K^*_i; \mathbf K^*_{-i}\}) = 0,
\label{palm-projectgrad:eq}
\end{equation}
\noindent where $\nabla_{\mathcal{G}_i(\mathbf K^*_{-i}), \eta} J_i(\{\mathbf K_i; \mathbf K^*_{-i}\})$ is the projected gradient of cost $J_i$ (\ref{Ji:eq}) onto the constraint set $\mathcal{G}_i$ (\ref{ith-constraint:eq}) for player $i$, $i=1,...,N$. Again, from \cite{hazan2017efficient}, we can write
\begin{align}
&\nabla_{\mathcal{G}_i(\mathbf K^*_{-i}), \eta} J_i(\{\mathbf K_i; \mathbf K^*_{-i}\}) \triangleq \nonumber\\
&\frac{1}{\eta}(\mathbf K_i - \Pi_{\mathcal{G}_i(\mathbf K^*_{-i})}[\mathbf K_i - \eta\nabla_{\mathbf K_i} J_i(\{\mathbf K_i, \mathbf K^*_{-i}\})])
\end{align}
\noindent where $\eta>0$ and the operator $\Pi_{\mathcal K}(\cdot)$ denotes projection onto the set $\mathcal K$. 

Similarly, in line 9 of Algorithm \ref{palm-game:alg}, a necessary condition for $\Phi_i$ to achieve its minimum is that the projected gradient $\nabla_{\mathcal{G}_i(\mathbf K^l_{-i}), \eta} J_i(\{\mathbf K_i; \mathbf K^l_{-i}\})=0$. 
Instead of seeking an exact equilibrium point as \gls{GNE}, we assume convergence of Algorithm \ref{palm-game:alg} when  this projected gradient is sufficiently small, which is a necessary condition for an approximate local equilibrium \cite{hazan2017efficient}. At iteration $l$, the $\hat{\mathbf K}_i$ in line 9 can be viewed as an approximation of $\Pi_{\mathcal{G}_i(\mathbf K^l_{-i})}[\mathbf K_i^{l-1} - \eta\nabla_{\mathbf K_i} J_i(\{\mathbf K_i^{l-1}, \mathbf K^l_{-i})\}]$. Thus, the norm of the projected gradient is proportional to $\vert| \mathbf K_i^{l} - \mathbf K_i^{l-1} \vert|$, implying that small values of  $\Delta \mathbf K_i^l$ and $\Delta \mathbf F_i^l$ indicate convergence of Algorithm \ref{palm-game:alg}. This is illustrated in Figure \ref{PALM-gne_delta:fig}. Moreover, we note that there is no theoretical guarantee for the existence of \gls{GNE} for the game in (\ref{ith-game:eq}). If a \gls{GNE} exists for the potential game (\ref{palm-pg:eq}) then this \gls{GNE} satisfies the necessary condition for the minimizer of (\ref{P}).

The main numerical complexity of Algorithms \ref{palm-alg1:alg} and \ref{palm-game:alg} is dominated by the $\mathbf K$-minimization step (Step 3 of Algorithm \ref{palm-alg1:alg} and line 9 of Algorithm \ref{palm-game:alg}), which has polynomial complexity on the number of variables in the feedback matrix \cite{gahinet1994lmi}.

\section{Conclusion}
\label{palm-concl:sec}

The \gls{PALM} method was exploited to solve the sparsity-constrained mixed $H_2/H_{\infty}$ control problem for multi-agent systems. First, a centralized social-optimization algorithm was investigated. Second, we developed noncooperative and potential games that have partially-distributed computation. The proposed algorithms were validated using an open-loop unstable  network dynamic system. It was demonstrated that the centralized \gls{PALM} method outperforms the \gls{GraSP}-based method for most sparsity levels, and converges both theoretically as well as in simulation results. Moreover, a best-response dynamics algorithm for the proposed games converges to an approximate \gls{GNE} point. The performance of the potential game for social optimization closely approximates that of the centralized algorithm.







\bibliographystyle{IEEEtran}
\bibliography{palmref}

\begin{thebibliography}{}

\bibitem[Bazaraa et~al., 2013]{bazaraa2013nonlinearsupp}
Bazaraa, M.~S., Sherali, H.~D., and Shetty, C.~M. (2013).
\newblock {\em Nonlinear programming: theory and algorithms}.
\newblock John Wiley \& Sons.

\bibitem[Bolte et~al., 2014]{bolte2014proximalsupp}
Bolte, J., Sabach, S., and Teboulle, M. (2014).
\newblock Proximal alternating linearized minimization or nonconvex and
  nonsmooth problems.
\newblock {\em Mathematical Programming}, 146(1-2):459--494.

\bibitem[Lin and Adetola, 2017]{lin2017cosupp}
Lin, F. and Adetola, V. (2017).
\newblock Co-design of sparse output feedback and row/column-sparse output
  matrix.
\newblock In {\em American Control Conference (ACC), 2017}, pages 4359--4364.
  IEEE.

\bibitem[Luenberger and Ye, 1984]{luenberger1984linearsupp}
Luenberger, D.~G. and Ye, Y. (1984).
\newblock {\em Linear and nonlinear programming}, volume~2.
\newblock Springer.

\bibitem[Netzer, 2016]{netzer2016realsupp}
Netzer, T. (2016).
\newblock Real algebraic geometry and its applications.
\newblock {\em arXiv preprint arXiv:1606.07284}.

\bibitem[Rautert and Sachs, 1997]{rautert1997computationalsupp}
Rautert, T. and Sachs, E.~W. (1997).
\newblock Computational design of optimal output feedback controllers.
\newblock {\em SIAM Journal on Optimization}, 7(3):837--852.

\end{thebibliography}


\begin{thebibliography}{10}
\providecommand{\url}[1]{#1}
\csname url@samestyle\endcsname
\providecommand{\newblock}{\relax}
\providecommand{\bibinfo}[2]{#2}
\providecommand{\BIBentrySTDinterwordspacing}{\spaceskip=0pt\relax}
\providecommand{\BIBentryALTinterwordstretchfactor}{4}
\providecommand{\BIBentryALTinterwordspacing}{\spaceskip=\fontdimen2\font plus
\BIBentryALTinterwordstretchfactor\fontdimen3\font minus
  \fontdimen4\font\relax}
\providecommand{\BIBforeignlanguage}[2]{{%
\expandafter\ifx\csname l@#1\endcsname\relax
\typeout{** WARNING: IEEEtran.bst: No hyphenation pattern has been}%
\typeout{** loaded for the language `#1'. Using the pattern for}%
\typeout{** the default language instead.}%
\else
\language=\csname l@#1\endcsname
\fi
#2}}
\providecommand{\BIBdecl}{\relax}
\BIBdecl

\bibitem{mihailo}
F.~Lin, M.~Fardad, and M.~R. Jovanovi{\'c}, ``Design of optimal sparse feedback
  gains via the alternating direction method of multipliers,'' \emph{IEEE
  Trans. on Aut. Ctrl.}, vol.~58, no.~9, pp. 2426--2431, 2013.

\bibitem{lian2017game}
F.~Lian, A.~Chakrabortty, and A.~Duel-Hallen, ``Game-theoretic multi-agent
  control and network cost allocation under communication constraints,''
  \emph{IEEE Journal on Selected Areas in Communications}, vol.~35, no.~2, pp.
  330--340, 2017.

\bibitem{javad}
N.~Monshizadeh, H.~L. Trentelman, and M.~K. Camlibel, ``Projection-based model
  reduction of multi-agent systems using graph partitions,'' \emph{IEEE Trans.
  on Control of Network Systems}, vol.~1, no.~2, pp. 145--154, 2014.

\bibitem{dorfler2014sparsity}
F.~D\"{o}rfler, M.~R. Jovanovi\'{c}, M.~Chertkov, and F.~Bullo,
  ``Sparsity-promoting optimal wide-area control of power networks,''
  \emph{IEEE Trans. on Power Systems}, vol.~29, no.~5, pp. 2281--2291, 2014.

\bibitem{lin2017co}
F.~Lin and V.~Adetola, ``Co-design of sparse output feedback and
  row/column-sparse output matrix,'' in \emph{American Control Conference
  (ACC), 2017}.\hskip 1em plus 0.5em minus 0.4em\relax IEEE, 2017, pp.
  4359--4364.

\bibitem{matni2016regularization}
N.~Matni and V.~Chandrasekaran, ``Regularization for design,'' \emph{IEEE
  Transactions on Automatic Control}, vol.~61, no.~12, pp. 3991--4006, 2016.

\bibitem{lidstrom2016optimal}
C.~Lidstr\"{o}m and A.~Rantzer, ``Optimal {$H_{\infty}$} state feedback for
  systems with symmetric and hurwitz state matrix,'' in \emph{American Control
  Conference}, 2016, pp. 3366--3371.

\bibitem{arastoo2016closed}
R.~Arastoo, M.~Bahavarnia, M.~V. Kothare, and N.~Motee, ``Closed-loop feedback
  sparsification under parametric uncertainties,'' in \emph{IEEE 55th
  Conference on Decision and Control}, 2016, pp. 123--128.

\bibitem{bahavarnia2017sparse}
M.~Bahavarnia and N.~Motee, ``Sparse memoryless {LQR} design for uncertain
  linear time-delay systems,'' \emph{IFAC-PapersOnLine}, vol.~50, no.~1, pp.
  10\,395--10\,400, 2017.

\bibitem{bahavarnia2017state}
M.~Bahavarnia, C.~Somarakis, and N.~Motee, ``State feedback controller
  sparsification via a notion of non-fragility,'' in \emph{Decision and Control
  (CDC), 2017 IEEE 56th Annual Conference on}.\hskip 1em plus 0.5em minus
  0.4em\relax IEEE, 2017, pp. 4205--4210.

\bibitem{seuken2008formal}
S.~Seuken and S.~Zilberstein, ``Formal models and algorithms for decentralized
  decision making under uncertainty,'' \emph{Autonomous Agents and Multi-Agent
  Systems}, vol.~17, no.~2, pp. 190--250, 2008.

\bibitem{ogren2004cooperative}
P.~Ogren, E.~Fiorelli, and N.~E. Leonard, ``Cooperative control of mobile
  sensor networks: Adaptive gradient climbing in a distributed environment,''
  \emph{IEEE Transactions on Automatic control}, vol.~49, no.~8, pp.
  1292--1302, 2004.

\bibitem{jungers2008bounded}
M.~Jungers, E.~B. Castelan, E.~R. De~Pieri, and H.~Abou-Kandil, ``Bounded nash
  type controls for uncertain linear systems,'' \emph{Automatica}, vol.~44,
  no.~7, pp. 1874--1879, 2008.

\bibitem{de2017finite}
N.~de~la Cruz and M.~Jimenez-Lizarraga, ``Finite time robust feedback nash
  equilibrium for linear quadratic games,'' \emph{IFAC-PapersOnLine}, vol.~50,
  no.~1, pp. 11\,794--11\,799, 2017.

\bibitem{basar85}
T.~Ba\c{s}ar and G.~J. Olster, \emph{Dynamic noncooperative game theory}.\hskip
  1em plus 0.5em minus 0.4em\relax SIAM, 1995, vol. 200.

\bibitem{mukaidani2009robust}
H.~Mukaidani, ``Robust guaranteed cost control for uncertain stochastic systems
  with multiple decision makers,'' \emph{Automatica}, vol.~45, no.~7, pp.
  1758--1764, 2009.

\bibitem{mukaidani2014h}
------, ``{H}$_2$/{H}$_{\infty}$ control problem for stochastic delay systems
  with multiple decision makers,'' in \emph{Decision and Control (CDC), 2014
  IEEE 53rd Annual Conference on}.\hskip 1em plus 0.5em minus 0.4em\relax IEEE,
  2014, pp. 2648--2653.

\bibitem{mukaidani2015stackelberg}
H.~Mukaidani and H.~Xu, ``Stackelberg strategies for stochastic systems with
  multiple followers,'' \emph{Automatica}, vol.~53, pp. 53--59, 2015.

\bibitem{mukaidani2016dynamic}
H.~Mukaidani, H.~Xu, and V.~Dragan, ``Dynamic games for markov jump stochastic
  delay systems,'' in \emph{Recent Results on Time-Delay Systems}.\hskip 1em
  plus 0.5em minus 0.4em\relax Springer, 2016, pp. 207--227.

\bibitem{vrabie2010integral}
D.~Vrabie and F.~Lewis, ``Integral reinforcement learning for online
  computation of feedback nash strategies of nonzero-sum differential games,''
  in \emph{Decision and Control (CDC), 2010 49th IEEE Conference on}.\hskip 1em
  plus 0.5em minus 0.4em\relax IEEE, 2010, pp. 3066--3071.

\bibitem{vamvoudakis2015non}
K.~G. Vamvoudakis, ``Non-zero sum {N}ash {Q}-learning for unknown deterministic
  continuous-time linear systems,'' \emph{Automatica}, vol.~61, pp. 274--281,
  2015.

\bibitem{song2017off}
R.~Song, F.~L. Lewis, and Q.~Wei, ``Off-policy integral reinforcement learning
  method to solve nonlinear continuous-time multiplayer nonzero-sum games,''
  \emph{IEEE transactions on neural networks and learning systems}, vol.~28,
  no.~3, pp. 704--713, 2017.

\bibitem{vamvoudakis2017game}
K.~G. Vamvoudakis, H.~Modares, B.~Kiumarsi, and F.~L. Lewis, ``Game
  theory-based control system algorithms with real-time reinforcement learning:
  How to solve multiplayer games online,'' \emph{IEEE Control Systems},
  vol.~37, no.~1, pp. 33--52, 2017.

\bibitem{Lian:2018aa}
F.~Lian, A.~Chakrabortty, F.~Wu, and A.~Duel-Hallen, ``Sparsity-constrained
  mixed {H}$_2$/{H}$_{\infty}$ control,'' in \emph{American Control Conference
  (ACC), 2018}, 2018.

\bibitem{bolte2014proximal}
J.~Bolte, S.~Sabach, and M.~Teboulle, ``Proximal alternating linearized
  minimization or nonconvex and nonsmooth problems,'' \emph{Mathematical
  Programming}, vol. 146, no. 1-2, pp. 459--494, 2014.

\bibitem{bahmani2013greedy}
S.~Bahmani, B.~Raj, and P.~T. Boufounos, ``Greedy sparsity-constrained
  optimization,'' \emph{The Journal of Machine Learning Research}, vol.~14,
  no.~1, pp. 807--841, 2013.

\bibitem{paccagnan2016distributed}
D.~Paccagnan, B.~Gentile, F.~Parise, M.~Kamgarpour, and J.~Lygeros,
  ``Distributed computation of generalized nash equilibria in quadratic
  aggregative games with affine coupling constraints,'' in \emph{Decision and
  Control (CDC), 2016 IEEE 55th Conference on}.\hskip 1em plus 0.5em minus
  0.4em\relax IEEE, 2016, pp. 6123--6128.

\bibitem{kami2008gradient}
Y.~Kami and E.~Nobuyama, ``A gradient method for the static output feedback
  mixed {H}$_2$/{H}$_{\infty}$ control,'' \emph{IFAC Proceedings Volumes},
  vol.~41, no.~2, pp. 7838--7842, 2008.

\bibitem{parikh2014proximal}
N.~Parikh, S.~Boyd \emph{et~al.}, ``Proximal algorithms,'' \emph{Foundations
  and Trends{\textregistered} in Optimization}, vol.~1, no.~3, pp. 127--239,
  2014.

\bibitem{Lian:aa}
F.~Lian, A.~Chakrabortty, and A.~Duel-Hallen, ``Supplementary materials for
  `{G}ame-theoretic mixed $h_2/h_{\infty}$ control with sparsity constraint for
  multi-agent networked control systems'.''

\bibitem{Lian:2019aa}
F.~Lian, ``Communication-cost-constrained algorithms and games for multi-agent
  control systems,'' Ph.D. dissertation, North Carolina State University,
  Raleigh, NC, USA, 2019.

\bibitem{luenberger1984linear}
D.~G. Luenberger and Y.~Ye, \emph{Linear and nonlinear programming}.\hskip 1em
  plus 0.5em minus 0.4em\relax Springer, 1984, vol.~2.

\bibitem{bazaraa2013nonlinear}
M.~S. Bazaraa, H.~D. Sherali, and C.~M. Shetty, \emph{Nonlinear programming:
  theory and algorithms}.\hskip 1em plus 0.5em minus 0.4em\relax John Wiley \&
  Sons, 2013.

\bibitem{saeki2006static}
M.~Saeki, ``Static output feedback design for {H}$_\infty$ control by descent
  method,'' in \emph{45th IEEE Conference on Decision and Control}, 2006, pp.
  5156--5161.

\bibitem{li2013designing}
N.~Li and J.~R. Marden, ``Designing games for distributed optimization,''
  \emph{IEEE Journal of Selected Topics in Signal Processing}, vol.~7, no.~2,
  pp. 230--242, 2013.

\bibitem{motee2008optimal}
N.~Motee and A.~Jadbabaie, ``Optimal control of spatially distributed
  systems,'' \emph{IEEE Trans. on Aut. Ctrl.}, vol.~53, no.~7, pp. 1616--1629,
  2008.

\bibitem{meyer2000matrix}
C.~D. Meyer, \emph{Matrix analysis and applied linear algebra}.\hskip 1em plus
  0.5em minus 0.4em\relax Siam, 2000, vol.~71.

\bibitem{cvx}
M.~Grant and S.~Boyd, ``{CVX}: Matlab software for disciplined convex
  programming, version 2.1,'' \url{http://cvxr.com/cvx}, Mar. 2014.

\bibitem{hazan2017efficient}
E.~Hazan, K.~Singh, and C.~Zhang, ``Efficient regret minimization in non-convex
  games,'' \emph{arXiv preprint arXiv:1708.00075}, 2017.

\bibitem{gahinet1994lmi}
P.~Gahinet, A.~Nemirovskii, A.~J. Laub, and M.~Chilali, ``The lmi control
  toolbox,'' in \emph{Proceedings of 1994 33rd IEEE Conference on Decision and
  Control}, vol.~3.\hskip 1em plus 0.5em minus 0.4em\relax IEEE, 1994, pp.
  2038--2041.

\end{thebibliography}

\title{Supplemental Materials for ``Game-Theoretic Mixed $H_2/H_{\infty}$ Control with
Sparsity Constraint for Multi-agent Networked Control Systems'' by Feier Lian, Aranya Chakrabortty, and Alexandra Duel-Hallen}

\makeatletter
\def\@seccntformat#1{\@ifundefined{#1@cntformat}%
   {\csname the#1\endcsname\quad}
   {\csname #1@cntformat\endcsname}
}
\makeatother

\newpage
\clearpage

\IEEEpeerreviewmaketitle

\setcounter{assumption}{0}
\setcounter{equation}{0}
\renewcommand{\theequation}{S\arabic{equation}}%

\setcounter{section}{0}
\makeatletter 
\newcommand{\section@cntformat}{S.\thesection:\ }
\makeatother

\section{Overview of Zoutendijk's method}
 \label{Zoute:appendix}

The Zoutendijk's method \citelatex{bazaraa2013nonlinearsupp} is an approach to constrained optimization, where an improving feasible direction is generated by solving a subproblem, usually a linear program. We briefly overview Zoutendijk's method for the case of nonlinear inequality constraints.

Consider the following constrained optimization problem:
 \begin{eqnarray}
 \label{palm-general:eq}
 \mbox{Minimize} &  &f(\mathbf x)\nonumber\\
 \mbox{s.t.} & &g_i(\mathbf x) \leq 0, ~i=1,...,m~, 
 \end{eqnarray}
 where $\mathbf x\in \mathbb{R}^{n\times 1}$ and $f(\mathbf x)$ and $g_i(\mathbf x)$ are differentiable at $\mathbf x$. At point $\mathbf x$, $I$ is the set of active constraint $I = \{i| g_i(\mathbf x)=0\}$. An improving feasible direction $\mathbf d$ can be found by the following linear programming problem \citelatex{bazaraa2013nonlinearsupp}:
 \begin{eqnarray}
 \underset{z, \mathbf d}{\mbox{Maximize}} & & z\nonumber\\
 \mbox{s. t.} & & \nabla f(\mathbf x)^T\mathbf d + z\leq 0,\nonumber\\
 & & \nabla g_i(\mathbf x)^T\mathbf d + z \leq 0 ~\forall i\in I, \nonumber\\
 & & -1\leq d_j \leq 1, ~\forall j=1,...,n,
 \label{Zou:eq}
 \end{eqnarray}
 where the {third} normalizing constraint prevents the optimal $z$ from approaching $\infty$. It was shown \citelatex{bazaraa2013nonlinearsupp} that if the optimal value of $z$, denoted as $z^*$, satisfies $z^*>0$, then $\mathbf d$ is an improving direction since $\mathbf d$ satisfies $\nabla f(\mathbf x)^T\mathbf d < 0$ and $\nabla g_i(\mathbf x)^T\mathbf d < 0 ~\forall i\in I$. Otherwise if $z^*=0$, then the current $\mathbf x$ is a Fritz John point \citelatex{bazaraa2013nonlinearsupp}, which satisfies the necessary condition for the local minimum of (\ref{palm-general:eq}).

\section{Definitions of Terms in Section II-C}

\begin{defn} (Lipschitz constant)
A function $f: \mathbb{R}^d \rightarrow \mathbb{R}$ with the gradient function $\nabla f$ is Lipschitz continuous with Lipschitz constant $L$ on $\mathcal S \in \mathbb{R}^d$ if $\vert| \nabla f(\mathbf x) - \nabla f(\mathbf y) \vert| \leq L\vert| \mathbf x - \mathbf y\vert|$ for all $\mathbf x, \mathbf y \in \mathcal S$ \citelatex{bazaraa2013nonlinearsupp}.
\end{defn}

\begin{defn} (Proper)
\label{proper:defn}
The function $\sigma: \mathcal{S} \rightarrow \mathbb{R}$ is a proper function if $\sigma(\mathbf x)>-\infty$ for all $\mathbf x \in \mathcal S$, and $\sigma(\mathbf x) < \infty$ for at least one point $\mathbf x\in \mathcal S$.
\end{defn}
\begin{defn}(Lower semicontinuous)
\label{lowersc:defn}
The function $\sigma: \mathcal{S} \rightarrow \mathbb{R}$ is lower semicontinuous at $\bar{\mathbf x}\in \mathcal{S}$ if for all $\epsilon>0$ there exists a $\delta$ such that $\mathbf x\in \mathcal{S}$ and $\vert| \mathbf x - \hat{\mathbf x}\vert| < \delta$ imply $\sigma(\mathbf x) - \sigma(\bar{\mathbf x}) > -\epsilon$.
\end{defn}

\section{Notation used in {Kurdyka-{\L}ojasiewicz (KL) Property}, employed in convergence analysis of Algorithm \ref{palm-alg1:alg}}

\begin{defn} (Distance.)
For any subset $\mathcal S \subset \mathbb{R}^d$ and any point $x \in \mathbb{R}^d$, the distance from $x$ to $\mathcal S$ is defined and denoted by
\begin{equation}
\mathrm{dist}(\mathbf{x}, \mathcal S) := \inf \{ \vert| \mathbf{y} - \mathbf{x}\vert| : \mathbf{y} \in \mathcal S\}.
\end{equation}
\end{defn}
\noindent When $\mathcal S = \emptyset$, we have $\mathrm{dist}(\mathbf{x}, \mathcal S) = \infty$ for all $\mathbf{x}$. 

Let $\eta \in [0, \infty]$. We denote by $\Phi_{\eta}$ the class of all concave and continuous functions $\varphi: [0, \eta) \rightarrow \mathbb{R}_+$ which satisfy the following conditions:\\
(i) $\varphi(0)=0$.\\
(ii) $\varphi$ has first-order continous derivative on $(0, \eta)$ and continous at $0$;\\
(iii) for all $s\in (0, \eta): \varphi'(s) > 0$.

\noindent For proper and lower semicontinous functions, the subdifferentials are defined below \citelatex{bolte2014proximalsupp}:
\begin{defn} (Subdifferentials)
\label{subdifferentials:defn}
Let $\sigma: \mathbb{R}^d\rightarrow (-\infty, \infty]$ be a proper and lower semicontinous function.\\
(i) For a given $\mathbf x \in \mathrm{dom}\, \sigma$, the Fr\'echet subdifferential of $\sigma$ at $x$, written as $\hat{\partial}\sigma(\mathbf x)$, is the set of all vectors $\mathbf u\in \mathbb{R}^d$ which satisfy
\begin{equation}
\lim_{\mathbf y\neq \mathbf x}\inf_{\mathbf y\rightarrow x}\frac{\sigma(\mathbf y) - \sigma(\mathbf x) - \langle \mathbf u, \mathbf y - \mathbf x\rangle}{\vert| \mathbf y-\mathbf x \vert|} \geq 0.
\end{equation}
When $\mathbf x\notin \mathrm{dom} \sigma$, we set $\hat{\partial}\sigma(\mathbf x)=\emptyset$.\\
(ii) The limiting subdifferential, or subdifferential, of $\sigma$ at $\mathbf x \in \mathbb{R}^d$, written $\partial \sigma(\mathbf x)$, is defined as
\begin{eqnarray}
\partial \sigma(\mathbf x)&:=&\left\{ \mathbf u\in \mathbb{R}^d: \exists \mathbf x^k \rightarrow \mathbf x, \sigma(\mathbf x^k) \rightarrow \sigma(\mathbf x)\right.\\ &&\left.~\mathrm{and}~ \mathbf u^k \in \hat{\partial}\sigma(\mathbf x^k) \rightarrow \mathbf u ~\mathrm{as}~ k\rightarrow \infty \right\}.
\end{eqnarray}
\end{defn}
\noindent Points whose subdifferentials contains $0$ are called \textit{(limiting-)critical points.}

\begin{defn} (Kurdyka-{\L}ojasiewicz (KL) Property)
\label{KL:defn}
Let $\sigma: \mathbb{R}^d \rightarrow (-\infty, +\infty]$ be proper and lower semicontinuous. \\
(i) The function $\sigma$ is said to have the \gls{KL} Property at $\bar{\mathbf{u}} \in \mathrm{dom} \,\partial \sigma := \{ \mathbf{u}\in \mathbb{R}^d: \partial \sigma(\mathbf{u}) \neq \emptyset\}$ if there exist $\eta\in (0, \infty]$, a neighorhood $\mathcal U$ of $\bar{\mathbf{u}}$ and a function $\varphi \in \Phi_{\eta}$, such that for all
\begin{equation}
\mathbf u \in \mathcal U \cap \left[ \sigma(\bar{\mathbf u}) < \sigma(\mathbf u) < \sigma(\bar{\mathbf u}) + \eta \right],
\end{equation}
\noindent the following inequality holds
\begin{equation}
\varphi'(\sigma(\mathbf u) - \sigma(\bar{\mathbf u}))\, \mathrm{dist}(0, \partial \sigma(\mathbf u)) \geq 1.
\end{equation}
(ii) If $\sigma$ satisfies the \gls{KL} property at each point of $\mathrm{dom} \, \partial\sigma$, then $\sigma$ is called a \gls{KL} function.
\end{defn}
It is shown in \citelatex{bolte2014proximalsupp} that \gls{KL} functions arise in many applications for optimization, in particular, semi-algebraic functions are \gls{KL} functions. The definitions for semi-algebraic function is given as follows.
\begin{defn}
(Semi-algebraic sets and functions). (i) A subset $\mathcal S \in \mathbb{R}^d$ is real semi-algebraic set if there exists a finite number of real polynomial functions $g_{ij}, h_{ij}: \mathbb{R}^d \rightarrow \mathbb{R}$ such that
\begin{equation}
\mathcal S = \cup_{j=1}^p\cap_{i=1}^q\left\{ \mathbf u \in \mathbb{R}^d: g_{ij}(\mathbf u) = 0 \mbox{ and } h_{ij}(\mathbf u) < 0\right\}.
\end{equation}
(ii) A function $h: \mathbb{R}^d \rightarrow (-\infty, +\infty]$ is called semi-algebraic if its graph
\begin{equation}
\left\{ (\mathbf u, t) \in \mathbb{R}^{d+1}: h(\mathbf u) = t\right\}
\end{equation}
is a semi-algebraic subset of $\mathbb{R}^{d+1}$.
\end{defn}

\section{Proof of Global Convergence of Algorithm \ref{palm-alg1:alg}}
\label{palm-centra-convg:sec}


In this section, we employ results from  \citelatex{lin2017cosupp,bolte2014proximalsupp} to analyze convergence of Algorithm \ref{palm-alg1:alg}.
To simplify notation, we define $\tilde g(\mathbf K) \triangleq J(\mathbf K) + g(\mathbf K)$, where $J(\mathbf K)$ is the performance index of $H_2$ cost (\ref{Jk:eq}) and $g(\mathbf K)$ is the indicator function for the $H_{\infty}$ constraint (\ref{palm-gK:eq}).



\begin{lem}
\label{lowsemi:lem}
$\tilde g: \mathbb{R}^{m\times p} \rightarrow (-\infty, \infty]$ and $f: \mathbb{R}^{m\times p} \rightarrow (-\infty, \infty]$ are proper and lower semicontinuous functions.  
\end{lem}

\begin{proof}
In problem (\ref{Pc}) the function $J(\mathbf K)$ is the LQR cost when using the feedback gain $\mathbf K$. Clearly $f(\mathcal K) > -\infty$, and $J(\mathbf K) < +\infty$ if $\mathbf K$ is stabilizing, thus function $J$ is proper. In addition, $J(\mathbf K)$ is continuous in $\mathbf K$ \citelatex{rautert1997computationalsupp}, and therefore lower semicontinuous.

 The function $g(\mathbf K)$ in (\ref{palm-gK:eq}) is the indicator function for the level set $\mathcal{K}(\gamma)$ in (\ref{palm-kappa:eq}), and thus can take either $0$ or $+\infty$, with $g(\mathbf K)=0$ whenever $\mathbf K \in \mathcal{K}(\gamma)$. Thus, $\tilde g(\mathbf K)$ is proper. In addition, $g(\mathbf K)$ is an indicator function of an open set, thus it is lower semicontinuous.
 Given $J$ and $g$ are both proper and lower semicontinuous, the summation $\tilde g = J + g$ is proper and lower semicontinuous. Similarly, $f(\mathbf F)$ in (\ref{palm-fF:eq}) is a proper function. Moreover, it is shown in \citelatex{bolte2014proximalsupp} that it is lower semicontinuous.
\end{proof}

\begin{lem}
\label{diff:lem}
$H: \mathbb{R}^{m\times p} \times \mathbb{R}^{m\times p}\rightarrow \mathbb{R}$ is a continuously differentiable function, i.e., $H\in C^1$.
\end{lem}

\begin{proof}
The gradient of $H(\mathbf K, \mathbf F)$ (\ref{gradH:eq}) is countinous in $\mathbf K, \mathbf F$. Thus, $H \in C^1$. 
\end{proof}

\begin{lem}
\label{lip:lem}
(i) $\inf_{\mathbb{R}^{m\times p}, \mathbb{R}^{m\times p}} \Phi > {-\infty}$, $\inf_{\mathbb{R}^{m\times p}} f > -\infty$, and $\inf_{\mathbb{R}^{m \times p}} \tilde g > -\infty$, where $\Phi$ is given by (\ref{Phi:eq}). \\
(ii) The partial gradient $\nabla_{\mathbf K} H(\mathbf K, \mathbf F)$ is globally Lipschitz with moduli $L_1(\mathbf F)$, that is \citelatex{bolte2014proximalsupp},
\begin{equation*}
\vert| \nabla_{\mathbf K} H(\mathbf K_1, \mathbf F) - \nabla_{\mathbf K} H(\mathbf K_2, \mathbf F) \vert| \leq L_1(\mathbf F) \vert| \mathbf K_1 - \mathbf K_2 \vert|.
\end{equation*}
\noindent Likewise, the partial gradient $\nabla_{\mathbf F} H(\mathbf K, \mathbf F)$ is globally Lipschitz with moduli $L_2(\mathbf K)$. \\
(iii) There exist bounds $\lambda_i^-$, $\lambda_i^+$, $i=1,2$ such that 
\begin{align}
\inf\{L_1(\mathbf F^k): k\in \mathbb{N}\} \geq \lambda_1^-, \inf\{L_2(\mathbf K^k):k\in \mathbb{N} \} \geq \lambda_2^- \nonumber\\
\sup \{L_1(\mathbf F^k): k\in \mathbb{N}\} \leq \lambda_1^+, \sup \{L_2(\mathbf K^k):k\in \mathbb{N} \} \leq \lambda_2^+
\end{align}
\noindent (iv) $\nabla H \triangleq (\nabla_{\mathbf K}H, \nabla_{\mathbf F}H)$ is Lipschitz continuous \citelatex{luenberger1984linearsupp} on bounded subsets of $\mathbb{R}^{m\times p} \times \mathbb{R}^{m \times p}$. That is, for each bounded subset $\mathcal{B}_1 \times \mathcal{B}_2$ of $\mathbb{R}^{m\times p} \times \mathbb{R}^{m \times p}$ there exists $M>0$, such that for all $(\mathbf K_i, \mathbf F_i) \in (\mathcal{B}_1, \mathcal{B}_2)$, 
\begin{eqnarray}
\label{global_Lip:eq}
\vert| \nabla_{\mathbf K}H(\mathbf K_1, \mathbf F_1) - \nabla_{\mathbf K}H(\mathbf K_2, \mathbf F_2) \vert|_F^2 \nonumber\\
+ \vert| \nabla_{\mathbf F}H(\mathbf K_1, \mathbf F_1) - \nabla_{\mathbf F}H(\mathbf K_2, \mathbf F_2) \vert|_F^2 \nonumber\\
\leq M  (\vert| \mathbf K_1 - \mathbf K_2 \vert|_F^2 + \vert| \mathbf F_1 - \mathbf F_2\vert|_F^2 )
\end{eqnarray}
\end{lem}

\begin{proof} 
(i)--(iv) are stated as assumptions in \citelatex{bolte2014proximalsupp}. We show that these assumptions hold for our sparsity-constrained mixed $H_2/H_{\infty}$ problem. It is easy to see that (i) holds since $f$ and $g$ are indicator functions. Since $J$ is the LQR performance index, $J(\mathbf K)>0$. Thus $\tilde g(\mathbf K) > -\infty$. In (\ref{H:eq}), $H(\mathbf K, \mathbf F) \geq 0$, so $\Phi(\mathbf K, \mathbf F)>-\infty$. 
Properties (ii) and (iii) require the partial gradient of $H$ to be globally Lipschitz, and the Lipschitz constant be upper and lower bounded, which is easy to verify since $L_1(\mathbf F^k)=L_2(\mathbf K^k)=\rho$ (\ref{gradH:eq}).
Property (iv) holds since the left-hand side of (\ref{global_Lip:eq}) can be expressed as: $LHS = 2\rho^2\vert| (\mathbf K_1 - \mathbf K_2) - (\mathbf F_1 - \mathbf F_2)\vert|_F^2 \leq 4\rho^2 (\vert| \mathbf K_1 - \mathbf K_2 \vert|_F^2 + \vert| \mathbf F_1 - \mathbf F_2 \vert|_F^2)$.
\end{proof}

\begin{assumption}
\label{KL:assumption}
Function $J$ is a semi-algebraic function \citelatex{bolte2014proximalsupp}.
\end{assumption}

\begin{lem}
\label{KL:lem}
The objective function $\Phi$ of (\ref{Pc}) is a Kurdyka-{\L}ojasiewicz (KL) function \citelatex{bolte2014proximalsupp}.
\end{lem}

\begin{rem}
A broad class of functions satisfy the semi-algebraic property, including polynomial functions, $\ell_0$-norm function and indicator function of positive semidefinite cones \citelatex{bolte2014proximalsupp}. 
The function $f(\mathbf F)$ is the indicator function for the semi-algebraic set $\{\mathbf F \vert \mathrm{card}(\mathbf F) \leq s \}$. Thus, function $f$ is semi-algebraic \citelatex{lin2017cosupp,bolte2014proximalsupp}. The function $g(\mathbf K)$ is the indicator function for the level set $\mathcal{K}(\gamma)$, which is approximated by the convex level set $\hat{\mathcal{K}}(\gamma_0)$, represented by the LMI (\ref{LMIsuff:eq}), and $\hat{\mathcal{K}}(\gamma_0)$ is a semi-algebraic set \citelatex{netzer2016realsupp}. The coupling function $H$ is polynomial, so it is semi-algebraic \citelatex{bolte2014proximalsupp}. Moreover, $J$ is a semi-algebraic function by Assumption \ref{KL:assumption}. Thus, each term of $\Phi$ is semi-algebraic, and since a finite sum of semi-algebraic functions is also semi-algebraic, $\Phi$ is semi-algebraic. It is shown in Theorem 5.1 in \citelatex{bolte2014proximalsupp} that a semi-algebraic function satisfies the KL property at any point in its domain. Thus, $\Phi$ is KL.

\end{rem}

It has been proved in \citelatex{bolte2014proximalsupp} that if Lemma \ref{lowsemi:lem}--\ref{lip:lem} hold, then
the sequence generated by PALM algorithm globally converges. In addition, if Lemma \ref{KL:lem} holds, the sequence converges to a critical point \citelatex{bolte2014proximalsupp} of $\Phi$. This confirms convergence of Algorithm \ref{palm-alg1:alg} to a sparsity-constrained mixed $H_2/H_{\infty}$ controller, which corresponds to a critical point of $\Phi$ under mild assumptions on the functions $J$ and $g$.

\bibliographystylelatex{apalike}
\bibliographylatex{palmref2}

\end{document}